\documentclass[12pt]{article}

\usepackage[a4paper,left=20mm,top=20mm,right=20mm,bottom=25mm]{geometry}

\usepackage{graphicx}
\usepackage{subcaption}
\usepackage{algorithm}
\usepackage[noend]{algpseudocode}
\usepackage{amsmath, amsfonts, amssymb, amsthm}
\usepackage{mathtools}
\usepackage{lscape}

\usepackage{color}
\usepackage[colorlinks]{hyperref}
\hypersetup{
	citecolor=blue,
   linkcolor=blue,
}
\usepackage[font=footnotesize,width=.85\textwidth,labelfont=bf]{caption}
\usepackage[mathcal]{euscript}
\usepackage{mathptmx} 

\usepackage{authblk}

\usepackage[colorlinks]{hyperref}
\hypersetup{
	citecolor=blue,
   linkcolor=red,
}
\usepackage[font=footnotesize,width=.85\textwidth,labelfont=bf]{caption}
\usepackage[mathcal]{euscript}
\usepackage{mathptmx}

\newcommand*\Let[2]{\State #1 $\gets$ #2}
\algrenewcommand\algorithmicrequire{\textbf{Precondition:}}
\algrenewcommand\algorithmicensure{\textbf{Postcondition:}}

\newcommand{\LO}{\textsf{LO}}
\newcommand{\UP}{\textsf{UP}}
\newcommand{\lo}{\textsf{lo}}
\newcommand{\up}{\textsf{up}}

\newcommand{\Gen}{\ensuremath{\mathbb{G}}}
\newcommand{\Spe}{\ensuremath{\mathbb{S}}}

\newcommand{\EHGT}{\ensuremath{\mathcal{E}}}
\newcommand{\Th}{\ensuremath{T_{\mathcal{\overline{E}}}}}
\newcommand{\lca}{\ensuremath{\operatorname{lca}}}
\newcommand{\spe}{\bullet}
\newcommand{\dpl}{\square}
\newcommand{\hgt}{\triangle}
\newcommand{\lea}{\odot}
\newcommand{\sT}{\sigma_{\Th}}
\newcommand{\tT}{\ensuremath{\tau_T}}
\newcommand{\ts}{\ensuremath{\tau_S}}
\newcommand{\mc}{\mathcal}
\newcommand{\mb}{\mathbb}

\providecommand{\keywords}[1]{\textbf{\textit{Keywords: }} #1}

\newtheorem{theorem}{Theorem}
\newtheorem{lemma}{Lemma}

\newtheorem{definition}{Definition}

\makeindex             

\begin{document}

\title{Forbidden Time Travel: Characterization of Time-Consistent Tree Reconciliation Maps}

\author[1,4]{Nikolai N{\o}jgaard}
\author[3]{Manuela Gei{\ss}}
\author[3,5,6,7]{Peter F.\ Stadler}
\author[4]{Daniel Merkle}
\author[8]{Nicolas Wieseke}
\author[1,2]{Marc Hellmuth}

\affil[1]{\footnotesize Dpt.\ of Mathematics and Computer Science, University of Greifswald, Walther-
  Rathenau-Strasse 47, D-17487 Greifswald, Germany \\
	\texttt{mhellmuth@mailbox.org}, 	\texttt{nikolai.nojgaard@uni-greifswald.de} }
\affil[2]{Saarland University, Center for Bioinformatics, Building E 2.1, P.O.\ Box 151150, D-66041 Saarbr{\"u}cken, Germany }
\affil[3]{Bioinformatics Group, Department of Computer Science; and
		    Interdisciplinary Center of Bioinformatics, University of Leipzig, \\
			 H{\"a}rtelstra{\ss}e 16-18, D-04107 Leipzig}
\affil[4]{Department of Mathematics and Computer Science,
		    University of Southern Denmark, Denmark }
\affil[5]{Max-Planck-Institute for Mathematics in the Sciences, \\
  Inselstra{\ss}e 22, D-04103 Leipzig}
\affil[6]{Inst.\ f.\ Theoretical Chemistry, University of Vienna, \\
  W{\"a}hringerstra{\ss}e 17, A-1090 Wien, Austria}
\affil[7]{Santa Fe Institute, 1399 Hyde Park Rd., Santa Fe, USA} 
\affil[8]{Parallel Computing and Complex Systems Group \\
  Department of Computer Science, 
  Leipzig University \\
  Augustusplatz 10, 04109, Leipzig, Germany}
\date{}
\normalsize

\maketitle

\abstract{ \noindent
 \emph{\textbf{Motivation:}} In the absence of horizontal gene transfer it is
  possible to reconstruct the history of gene families from empirically
  determined orthology relations, which are equivalent to
  \emph{event-labeled} gene trees. Knowledge of the event labels
  considerably simplifies the problem of reconciling a gene tree $T$ with a
  species trees $S$, relative to the reconciliation problem without prior
  knowledge of the event types. It is well-known that optimal
  reconciliations in the unlabeled case may violate time-consistency and
  thus are not biologically feasible.  Here we investigate the mathematical
  structure of the event labeled reconciliation problem with horizontal
  transfer.  \\
	\emph{\textbf{Results:}} We investigate the issue of time-consistency
  for the event-labeled version of the reconciliation problem, provide a
  convenient axiomatic framework, and derive a complete characterization of
  time-consistent reconciliations. This characterization depends on
    certain weak conditions on the event-labeled gene trees that reflect
    conditions under which evolutionary events are observable at least in
    principle. We give an $\mathcal{O}(|V(T)|\log(|V(S)|))$-time algorithm
  to decide whether a time-consistent reconciliation map exists. It does
  not require the construction of explicit timing maps, but relies entirely
  on the comparably easy task of checking whether a small auxiliary graph
  is acyclic.  \\
\emph{\textbf{Significance:}} The combinatorial characterization of
  time consistency and thus biologically feasible reconciliation is an
  important step towards the inference of gene family histories with
  horizontal transfer from orthology data, i.e., without presupposed gene
  and species trees. The fast algorithm to decide time consistency is 
  useful in a broader context because it constitutes an attractive
  component for all tools that address tree reconciliation problems.
}

\bigskip
\noindent

\sloppy

\section{Introduction}

Modern molecular biology describes the evolution of species in terms of the
evolution of the genes that collectively form an organism's genome. In this
picture, genes are viewed as atomic units whose evolutionary history
\emph{by definition} forms a tree. The phylogeny of species also forms a
tree. This species tree is either interpreted as a consensus of the gene
trees or it is inferred from other data. An interesting formal manner to
define a species tree independent of genes and genetic data is discussed
e.g.\ in \cite{Dress:10}. In this contribution, we assume that gene and
species trees are given independently of each other. The relationship
between gene and species evolution is therefore given by a reconciliation
map that describes how the gene tree is embedded in the species tree: after
all, genes reside in organisms, and thus at each point in time can be
assigned to a species.

From a formal point of view, a reconciliation map $\mu$ identifies vertices
of a gene tree with vertices and edges in the species tree in such a way
that (partial) ancestor relations given by the genes are preserved by
$\mu$. Vertices in the species tree correspond to speciation events. Since
in this situation genes are faithfully transmitted from the parent species
into both (all) daughter species, some of the vertices in the gene tree
correspond to speciation events. Other important events considered here are
gene duplications, in which two copies of a gene keep residing in the same
species, and horizontal gene transfer events (HGT). Here, the original
remains in the parental species, while the offspring copy ``jumps'' into a
different branch of the species tree. It is customary to define pairwise
relations between genes depending on the event type of their last common
ancestor \cite{Fitch2000,HSW:16,HW:16b}.

Most of the literature on this topic assumes that both the gene tree and
the species tree are known. The aim is then to find a mapping of the gene
tree $T$ into the species tree $S$ and, at least implicitly, an
event-labeling on the vertices of the gene tree $T$. Here we take a
different point of view and assume that $T$ and the types of evolutionary
events on $T$ are known. This setting has ample practical relevance because
event-labeled gene trees can be derived from the pairwise orthology
relation \cite{Hellmuth:15a,HW:16b}. These relations in turn can be
estimated directly from sequence data using a variety of algorithmic
approaches that are based on the pairwise best match criterion and hence do not
require any \emph{a priori} knowledge of the topology of either the gene
tree or the species tree, see e.g.\
\cite{Roth:08,Altenhoff:09,Lechner:14,Altenhoff:16}.

Genes that share a common origin (homologs) can be classified into
orthologs, paralogs, and xenologs depending whether they originated by a
speciation, duplication or horizontal gene transfer (HGT) event
\cite{Fitch2000,HW:16b}.  Recent advances in mathematical phylogenetics
\cite{HHH+13,HSW:16,Hellmuth:15a} have shown that the knowledge of these
event-relations (orthologs, paralogs and xenologs) suffices to construct
event-labeled gene trees and, in some case, also a species tree.

Conceptually, both the gene tree and species tree are associated with a
timing of each event. Reconciliation maps must preserve this timing
information because there are \emph{biologically infeasible} event
  labeled gene trees that cannot be reconciled with any species tree.  In
the absence of HGT, biologically feasibility can be characterized in terms
of certain triples (rooted binary trees on three leaves) that are displayed
by the gene trees \cite{HHH+12}. In contrast, the timing information
  must be taken into account explicitly in the presence of HGT.  In other
words, there are gene trees with HGT that can be mapped to species
  trees only in such a way that some genes travels back in time.

There have been several attempts in the literature to handle this issue,
see e.g.\ \cite{Doyon2011} for a review. In \cite{MeMi05,Ch98} a
\emph{single} HGT adds timing constraints to a time map for a
reconciliation to be found.  Time-consistency is then subsequently defined
based on the existence of a topological order of the digraph reflecting all
the time constraints. In \cite{THL:11} NP-hardness was shown for finding a
parsimonious time-consistent reconciliation based on a definition for
time-consistency that essentially is based on considering \emph{pairs} of
HGTs. This subtle modification, however, makes the definition for
time-consistency only a necessary but not sufficient definition.  Fig.\ 
\ref{fig:Mu}(left), for example, shows a biologically infeasible example
that nevertheless satisfies the conditions for time-consistency given in
\cite{THL:11}, but not those of \cite{MeMi05,Ch98}. Different algorithmic
approaches for tackling time-consistency exist \cite{Doyon2011} such as the
inclusion of time-zones known for specific evolutionary events. It is worth
noting that \emph{a posteriori} modifications of time-inconsistent
solutions will in general violate parsimony \cite{MeMi05}.

Here, we introduce an axiomatic framework for time-consistent
reconciliation maps and characterize for given event-labeled gene trees and
species trees whether there exists a time-consistent reconciliation map. We
provide an algorithm that constructs a time-consistent reconciliation map
if one exists.

\section{Notation and Preliminaries}
\label{sec:prelim}


We consider \emph{rooted trees $T=(V,E)$ (on $L_T$)} with root
$\rho_T \in V$ and leaf set $L_T\subseteq V$.  A vertex ${v}\in V$ is
called a \emph{descendant} of ${u}\in V$, ${v \preceq_T u}$, and ${u}$ is an
\emph{ancestor} of ${v}$, ${u \succeq_T v}$, if ${u}$ lies on the path from
$\rho_T$ to ${v}$. 
As usual, we write ${v \prec_T u}$ and ${u \succ_T v}$ to
mean ${v \preceq_T u}$ and $u\ne v$. The partial order $\succeq_T$ is known
as the \emph{ancestor order} of $T$; the root is the unique maximal element
w.r.t\ $\succeq_T$. If $u \preceq_T v$ or $v \preceq_T u$ then $u$ and $v$
are \emph{comparable} and otherwise, \emph{incomparable}.  We consider
edges of rooted trees to be directed away from the root, that is, the
notation for edges $(u,v)$ of a tree is chosen such that $u\succ_T v$.  If
$(u,v)$ is an edge in $T$, then $u$ is called \emph{parent} of $v$ and $v$
\emph{child} of $u$. It will be convenient for the discussion below to
extend the ancestor relation $\preceq_T$ on $V$ to the union of the edge
and vertex sets of $T$. More precisely, for the edge $e=(u,v)\in E$ we put
$x \prec_T e$ if and only if $x\preceq_T v$ and $e \prec_T x$ if and only
if $u\preceq_T x$. For edges $e=(u,v)$ and $f=(a,b)$ in $T$ we put
$e\preceq_T f$ if and only if $v \preceq_T b$. For $x\in V$, we write
$L_T(x):=\{y\in L_T \mid y\preceq_T x\}$ for the set of leaves in the
subtree $T(x)$ of $T$ rooted in $x$.

For a non-empty subset of leaves $A\subseteq L$, we define $\lca_T(A)$, or
the \emph{least common ancestor of $A$}, to be the unique
$\preceq_T$-minimal vertex of $T$ that is an ancestor of every vertex in
$A$. In case $A=\{u,v \}$, we put $\lca_T(u,v):=\lca_T(\{u,v\})$. We have
in particular $u=\lca_T(L_T(u))$ for all $u\in V$.  We will also frequently
use that for any two non-empty vertex sets $A,B$ of a tree, it holds that
$\lca(A\cup B) = \lca(\lca(A),\lca(B))$.

A \emph{phylogenetic tree} is a rooted tree such that no interior vertex in
$v\in V\setminus L_T$ has degree two, except possibly the root If $L_T$
corresponds to a \emph{set of genes} $\Gen$ or \emph{species} $\Spe$, we
call a phylogenetic tree on $L_T$ \emph{gene tree} or \emph{species tree},
respectively. In this contribution we will \textbf{not} restrict the gene
or species trees to be binary, although this assumption is made implicitly
or explicitly in much of the literature on the topic. The more general
setting allows us to model incomplete knowledge of the exact gene or
species phylogenies. Of course, all mathematical results proved here also
hold for the special case of binary phylogenetic trees.

In our setting a gene tree $T=(V,E)$ on $\Gen$ is equipped with an
\emph{event-labeling} map $t:V\cup E\to I\cup \{0,1\}$ with
$I=\{\bullet,\square,\triangle, \odot\}$ that assigns to each interior
vertex $v$ of $T$ a value $t(v)\in I$ indicating whether $v$ is a
speciation event ($\bullet$), duplication event ($\square$) or HGT event
($\triangle$). It is convenient to use the special label $\odot$ for the
leaves $x$ of $T$. Moreover, to each edge $e$ a value $t(e)\in \{0,1\}$ is
added that indicates whether $e$ is a \emph{transfer edge} ($1$) or not
($0$). Note, only edges $(x,y)$ for which $t(x)=\triangle$ might be labeled
as transfer edge. We write $\EHGT = \{e\in E\mid t(e)=1\}$ for the set of
transfer edges in $T$.  We assume here that all edges
  labeled ``$0$'' transmit the genetic material vertically, that is, from
  an ancestral species to its descendants.

We remark that the restriction $t_{|V}$ of $t$ to the vertex set $V$
coincides with the ``symbolic dating maps'' introduced in
\cite{Boeckner:98}; these have a close relationship with cographs
\cite{HHH+13,HW:15,HW:16a}.  Furthermore, there is a map $\sigma:\Gen\to
\Spe$ that assigns to each gene the species in which it resides. The set
$\sigma(M)$, $M\subseteq \Gen$, is the set of species from which the genes
$M$ are taken. We write $(T;t,\sigma)$ for the gene tree $T=(V,E)$ with
event-labeling $t$ and corresponding map $\sigma$.

Removal of the transfer edges from $(T;t,\sigma)$ yields a forest
$\Th\coloneqq (V,E\setminus \EHGT)$ that inherits the ancestor order on
its connected components, i.e., $\preceq_{\Th}$ iff $x\preceq_{T}y$ and
$x,y$ are in same subtree of $\Th$ \cite{THL:11}. Clearly
$\preceq_{\Th}$ uniquely defines a root for each subtree and the set of
descendant leaf nodes $L_{\Th}(x)$.

In order to account for duplication events that occurred before the first
speciation event, we need to add an extra vertex and an extra edge
``above'' the last common ancestor of all species in the species tree
$S=(V,E)$.  Hence, we add an additional vertex to $V$ (that is now the new
root $\rho_S$ of $S$) and the additional edge $(\rho_S,\lca_S(\Spe))$ to
$E$. Strictly speaking $S$ is not a phylogenetic tree in the usual sense,
however, it will be convenient to work with these augmented trees.  For
simplicity, we omit drawing the augmenting edge $(\rho_S,\lca_S(\Spe))$ in
our examples.

\section{Observable Scenarios} 

The true history of a gene family, as it is considered here, is an
arbitrary sequence of speciation, duplication, HGT, and gene loss
events. The applications we envision for the theory developed, here,
however assume that the gene tree and its event labels are inferred from
(sequence) data, i.e., $(T;t,\sigma)$ is restricted to those labeled trees
that can be constructed at least in principle from observable data. The
issue here are gene losses that may complete eradicate the information on
parts of the history. Specifically, we require that $(T;t,\sigma)$
satisfies the following three conditions:

\begin{description}
\item[(O1)] Every internal vertex $v$ has degree at least $3$, except
  possibly the root which has degree at least $2$.
\item[(O2)] Every HGT node has at least one transfer edge, $t(e)=1$, and at
  least one non-transfer edge, $t(e)=0$;
\item[(O3)] \emph{\textbf{(a)}} If $x$ is a speciation vertex, then there
  are at least two distinct children $v,w$ of $x$ such that the species $V$
  and $W$ that contain $v$ and $w$, resp., are incomparable in $S$.\\
  \emph{\textbf{(b)}} If $(v,w)$ is a transfer edge in $T$, then the
  species $V$ and $W$ that contain $v$ and $w$, resp., are
  incomparable in $S$.
\end{description}

Condition (O1) ensures that every event leaves a historical trace in the
sense that there are at least two children that have survived in at least
two of its subtrees. If this were not the case, no evidence would be left
for all but one descendant tree, i.e., we would have no evidence that event
$v$ ever happened. We note that this condition was used e.g.\ in
\cite{HHH+12} for scenarios without HGT.  Condition (O2) ensures that for
an HGT event a historical trace remains of both the transferred and the
non-transferred copy. If there is no transfer edge, we have no evidence to
classify $v$ as a HGT node. Conversely, if all edges were transfers, no
evidence of the lineage of origin would be available and any reasonable
inference of the gene tree from data would assume that the gene family was
vertically transmitted in at least one of the lineages in which it is
observed. In particular, Condition (O2) implies that for each internal
vertex there is a path consisting entirely of non-transfer edges to some
leaf. This excludes in particular scenarios in which a gene is transferred
to a different ``host'' and later reverts back to descendants of the
original lineage without any surviving offspring in the intermittent host
lineage.  Furthermore, a speciation vertex $x$ cannot be observed from data
if it does not ``separate'' lineages, that is, there are two leaf
descendants of distinct children of $x$ that are in distinct
species. However, here we only assume to have the weaker Condition (O3.a)
which ensures that any ``observable'' speciation vertex $x$ separates at
least locally two lineages. In other words, if all children of $x$ would be
contained in species that are comparable in $S$ or, equivalently, in the
same lineage of $S$, then there is no clear historical trace that justifies
$x$ to be a speciation vertex.  In particular, most-likely there are two
leaf descendants of distinct children of $x$ that are in the same species
even if only $\Th$ is considered. Hence, $x$ would rather be classified as
a duplication than as a speciation upon inference of the event labels from
actual data.
Analogously, if $(v,w)\in \EHGT$ then $v$ signifies the transfer event
itself but $w$ refers to the next (visible) event in the gene tree
$T$. Given that $(v,w)$ is a HGT-edge in the observable part, in a ``true
history'' $v$ is contained in a species $V$ that transmits its genetic
material (maybe along a path of transfers) to a contemporary species $Z$
that is an ancestor of the species $W$ containing $w$. Clearly, the latter
allows to have $V\succeq_S W$ which happens if the path of transfers points
back to the descendant lineage of $V$ in $S$.  In this case the transfer
edge $(v,w)$ must be placed in the species tree such that $\mu(v)$ and
$\mu(w)$ are comparable in $S$.  However, then there is no evidence that
this transfer ever happened, and thus $v$ would be rather classified as
speciation or duplication vertex.

\begin{figure}[tbp]
  \begin{center}
    \includegraphics[width=.9\textwidth]{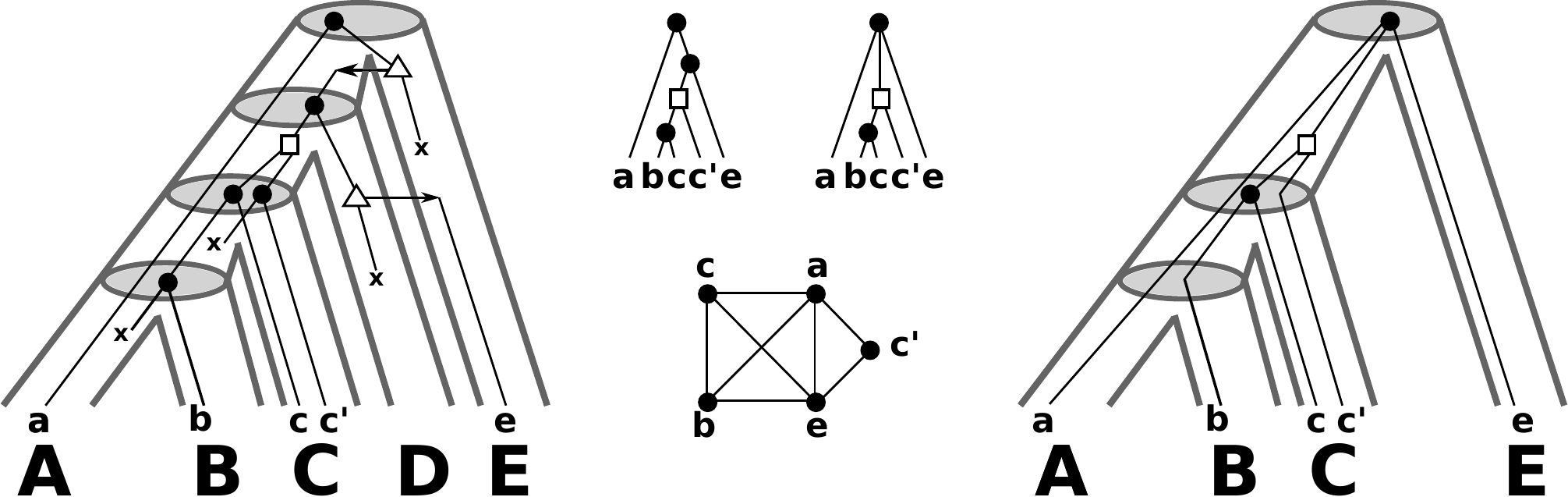}
  \end{center}
  \caption{\emph{Left:} A ``true'' evolutionary scenario for a gene tree
    with leaf set $\Gen$ evolving along the tube-like species trees is
    shown. The symbol ``x'' denotes losses.  All speciations along the path
    from the root $\rho_T$ to the leaf $a$ are followed by losses and we
    omit drawing them.  
    \newline 
    \emph{Middle:} The observable gene tree is shown in the upper-left.
    The orthology graph $G = (\Gen,E)$ (edges are placed between genes
    $x,y$ for which $t(\lca(x,y)) = \bullet$) is drawn in the lower
    part. This graph is a cograph and the corresponding \emph{non-binary}
    gene tree $T$ on $\Gen$ that can be constructed from such data is given
    in the upper-right part (cf.\ \cite{HHH+13,HSW:16,HW:16b} for further
    details). 
    \newline 
    \emph{Right:} Shown is species trees $S$ on $\Spe =\sigma(\Gen)$ with
    reconciled gene tree $T$.  The reconciliation map $\mu$ for $T$ and $S$
    is given implicitly by drawing the gene tree $T$ within $S$.  Note,
    this reconciliation is not consistent with DTL-scenarios
    \cite{THL:11,BAK:12}.  A DTL-scenario would require that the
    duplication vertex and the leaf $a$ are incomparable in $S$, see
    Appendix \ref{sec:app:dtl} for further details.}
  \label{fig:counterDTL}
\end{figure}

We show in Appendix~\ref{sec:app:obs} that (O1), (O2) and (O3) imply Lemma
\ref{lem:part_gen} as well as two important properties ($\Sigma$1) and
($\Sigma$2) of event labeled species trees that play a crucial role for
the results reported here.

\begin{lemma}
  Let $\mathcal{T}_1, \dots, \mathcal{T}_k$ be the connected components of
  $\Th$ with roots $\rho_1, \dots, \rho_k$, respectively.  If (O1), (O2) and (O3)
  holds, then, $\{L_{\Th}(\rho_1), \dots, L_{\Th}(\rho_k)\}$ forms a
  partition of $\Gen$.
  \label{lem:part_gen}
\end{lemma}
\begin{description}
\item[($\Sigma$1)] If $t(x)=\bullet$ then there are distinct children $v$,
  $w$ of $x$ in $T$ such that $\sigma(L_{\Th}(v))\cap \sigma(L_{\Th}(w)) =
  \emptyset$.
  \label{des:spec}
\end{description}

Intuitively, ($\Sigma$1) is true because within a component $\Th$ no
genetic material is exchanged between non-comparable nodes. Thus, a gene
separated in a speciation event necessarily ends up in distinct species in
the absence of horizontal transfer.  It is important to note that we do not
require the converse: $\sigma(L_{\Th}(y))\cap\sigma(L_{\Th}(y'))=\emptyset$
does \textbf{not} imply $t(\lca_T(L_{\Th}(y)\cup L_{\Th}(y'))=\bullet$,
that is, the last common ancestor of two sets of genes from different
species is not necessarily a speciation vertex.

Now consider a transfer edge $(v,w) \in \EHGT$, i.e., $t(v)=\triangle$.
Then $\Th(v)$ and $\Th(w)$ are subtrees of distinct connected components of
$\Th$. Since HGT amounts to the transfer of genetic material \emph{across}
distinct species, the genes $v$ and $w$ must be contained in distinct
species $X$ and $Y$, respectively. Since no genetic material is 
transferred between contemporary species $X'$ and $Y'$ in $\Th$, where
$X'$ and $Y'$ is a descendant of $X$ and $Y$, respectively we derive  

\begin{description}
\item[($\Sigma$2)] If $(v,w) \in \EHGT$ then $\sigma(L_{\Th}(v))\cap
  \sigma(L_{\Th}(w)) = \emptyset$.
  \label{des:HGT}
\end{description}

A more formal proof for Lemma \ref{lem:part_gen}, ($\Sigma$1) and ($\Sigma$2) is given in Appendix~\ref{sec:app:obs}.  From here on
we simplify the notation a bit and write $\sT(u):=\sigma(L_{\Th}(u))$.

\section{Time-Consistent Reconciliation Maps}

The problem of reconciliation between gene trees and species tree is
  formalized in terms of so-called DTL-scenarios in the literature
  \cite{THL:11,BAK:12}. This framework, however, usually assumes that the
  event labels $t$ on $T$ are unknown, while a species tree $S$ is
  given. The ``usual'' DTL axioms, furthermore, explicitly refer to binary,
  fully resolved gene and species trees. We therefore use a different axiom
  set here that is a natural generalization of the framework introduced in
  \cite{HHH+12} for the HGT-free case:
\begin{definition} \label{def:mu} 
  Let $T=(V,E)$ and $S=(W,F)$ be phylogenetic trees on $\Gen$ and 
  $\Spe$, 
  resp., $\sigma:\Gen\to\Spe$ the assignment of genes to species and 
  $t:V\cup E\to \{\bullet,\square,\triangle,\odot\} \cup \{0,1\}$ an event
  labeling on $T$. A map $\mu:V\to W\cup F$ is a \emph{reconciliation map} 
  if for all $v\in V$ it holds that: 
  \begin{description}
  \item[(M1)] \emph{Leaf Constraint.}  If $t(v)=\odot$, then
    $\mu(v)=\sigma(v)$.
  \item[(M2)] \emph{Event Constraint.}
    \begin{itemize}
    \item[(i)] If $t(v)=\bullet$, then $\mu(v) = \lca_S(\sT(v))$.
    \item[(ii)] If $t(v) \in \{\square, \triangle\}$, then $\mu(v)\in F$.
    \item[(iii)] If $t(v)=\triangle$ and $(v,w)\in \EHGT$, then $\mu(v)$
      and $\mu(w)$ are incomparable in $S$.
    \end{itemize}
  \item[(M3)] \emph{Ancestor Constraint.}		\\
    Suppose $v,w\in V$ with $v\prec_{\Th} w$.  
    \begin{itemize}
    \item[(i)] If $t(v),t(w)\in \{\square, \triangle\}$, then
      $\mu(v)\preceq_S \mu(w)$,
    \item[(ii)] otherwise, i.e., at least one of $t(v)$ and $t(w)$ is a
      speciation $\bullet$, $\mu(v)\prec_S\mu(w)$.
    \end{itemize}
  \end{description}
  We say that $S$ is a \emph{species tree for $(T;t,\sigma)$} if a
  reconciliation map $\mu:V\to W\cup F$ exists.
\end{definition}                
In Appendix~\ref{sec:app:dtl} we provide a translation of the DTL axioms to
  the notation used here and show that, for binary trees, they are
  equivalent to Definition~\ref{def:mu}. In Figure \ref{fig:counterDTL}
   an example of a biologically plausible reconciliation of
  non-binary trees that is valid w.r.t.\ Definition~\ref{def:mu} is shown, however,  
	it does not satisfy the conditions of a DTL-scenario.

\begin{figure}[tbp]
  \begin{center}
    \includegraphics[width=0.85\textwidth]{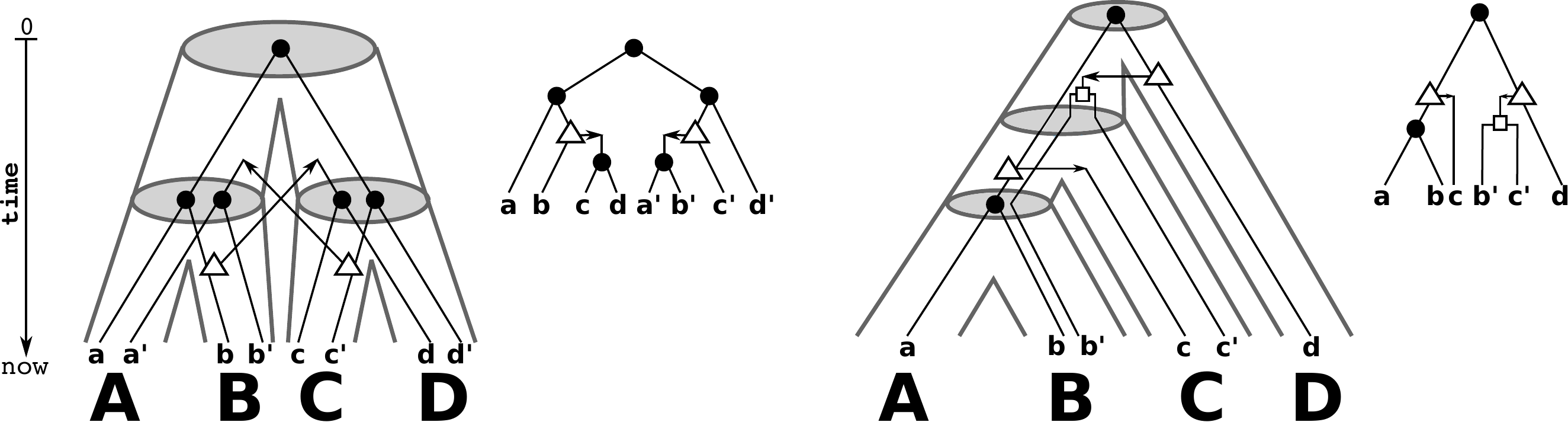}
  \end{center}
  \caption{Shown are two (tube-like) species trees with reconciled gene trees.  The
    reconciliation map $\mu$ for $T$ and $S$ is given implicitly by drawing
    the gene tree (upper right to the respective species tree) within the
    species tree.  \newline 
	In the left example, the map $\mu$ is unique.  However, although
    $\mu$ satisfies the conditions of Def.\ \ref{def:mu}, it is not
    time-consistent. 
	This example, however, is time-consistent 
	by the definition in \cite{THL:11}. Thus, 
 	time-consistency as given in
	\cite{THL:11} is only necessary but not sufficient. 
		In the right example, $\mu$ is time-consistent. }
	\label{fig:Mu}
\end{figure}

Condition (M1) ensures that each leaf of $T$, i.e., an extant gene in
$\Gen$, is mapped to the species in which it resides. Conditions (M2.i) and
(M2.ii) ensure that each inner vertex of $T$ is either mapped to a vertex
or an edge in $S$ such that a vertex of $T$ is mapped to an interior vertex
of $S$ if and only if it is a speciation vertex.  Condition (M2.i) might
seem overly restrictive, an issue to which we will return below.  Condition
(M2.iii) satisfies condition (O3) and 
maps the vertices of a transfer edge in a way that they are
incomparable in the species tree, since a HGT occurs between distinct
(co-existing) species. 
It becomes void in the absence of HGT; thus Definition \ref{def:mu} reduces
to the definition of reconciliation maps given in \cite{HHH+12} for the
HGT-free case.  Importantly, condition (M3) refers only to the connected
components of $\Th$ since comparability w.r.t.\ $\prec_{\Th}$ implies that
the path between $x$ and $y$ in $T$ does not contain transfer edges. It
ensures that the ancestor order $\preceq_T$ of $T$ is preserved along all
paths that do not contain transfer edges.

We will make use of the following bound that effectively restricts how
close to the leafs the image of a vertex in the gene tree can be located.  
\begin{lemma}
  If $\mu: (T;t,\sigma)\to S$ satisfies (M1) and (M3), then
  $\mu(u)\succeq_S \lca_S(\sT(u))$ for any $u\in V(T)$.
\label{lem:cond-mu2}
\end{lemma}
\begin{proof}
  If $u$ is a leaf, then by Condition (M1) $\mu(u)=\sigma(u)$ and we are
  done. Thus, let $u$ be an interior vertex. By Condition (M3), $z
  \preceq_S\mu(u)$ for all $z\in \sT(u)$. Hence, if $\mu(u)\prec_S
  \lca_S(\sT(u))$ or if $\mu(u)$ and $\lca_S(\sT(u)))$ are incomparable in
  $S$, then there is a $z\ \in \sT(u)$ such that $z$ and $\mu(u)$ are
  incomparable; contradicting (M3).
\end{proof}
Condition (M2.i) implies in particular the weaker property ``(M2.i') if
$t(v)=\bullet$ then $\mu(v)\in W$''. In the light of
Lemma~\ref{lem:cond-mu2}, $\mu(v)=\lca_S(\sT(v))$ is the lowest possible
choice for the image of a speciation vertex.  Clearly, this restricts the
possibly exponentially many reconciliation maps for which
$\mu(v)\succ_S\lca_S(\sT(v))$ for speciation vertices $v$ is allowed to
only those that satisfy (M2.i). However, the latter is justified by the
observation that if $v$ is a speciation vertex with children $u,w$, then
there is only one unique piece of information given by the gene tree to
place $\mu(v)$, that is, the unique vertex $x$ in $S$ with children $y,z$
such that $\sT(u) \subseteq L_S(y)$ and $\sT(w) \subseteq L_S(z)$.  The
latter arguments easily generalizes to the case that $v$ has more than two
children in $T$.  Moreover, any \emph{observable} speciation node
$v'\succ_T v$ closer to the root than $v$ must be mapped to a node
ancestral to $\mu(v)$ due to (M3.ii). Therefore, we require 
$\mu(v) =x = \lca_S(\sT(v))$ here. 

If $S$ is a species tree for the gene tree $(T,t,\sigma)$ then there is no
freedom in the construction of a reconciliation map $\mu$ on the set
$\{x\in V(T)\mid t(x)\in\{\bullet, \odot\}\}$.  The duplication and HGT
vertices of $T$, however, can be placed differently. As a consequence there
is a possibly exponentially large set of reconciliation maps from
$(T,t,\sigma)$ to $S$. 

From a biological point of view, however, the notion of reconciliation used
so far is too weak. In the absence of HGT, subtrees evolve independently and
hence, the linear order of points along each path from root to leaf is
consistent with a global time axis. This is no longer true in the presence
of HGT events, because HGT events imply additional time-consistency
conditions. These stem from the fact that the appearance of the HGT copy in
a distant subtree of $S$ is concurrent with the HGT event. To investigate
this issue in detail, we introduce time maps and the notion of
time-consistency, see Figures \ref{fig:Mu} -- \ref{fig:nonU-Mu2}
for illustrative examples.

\begin{figure}[tbp]
  \begin{center}
    \includegraphics[width=0.85\textwidth]{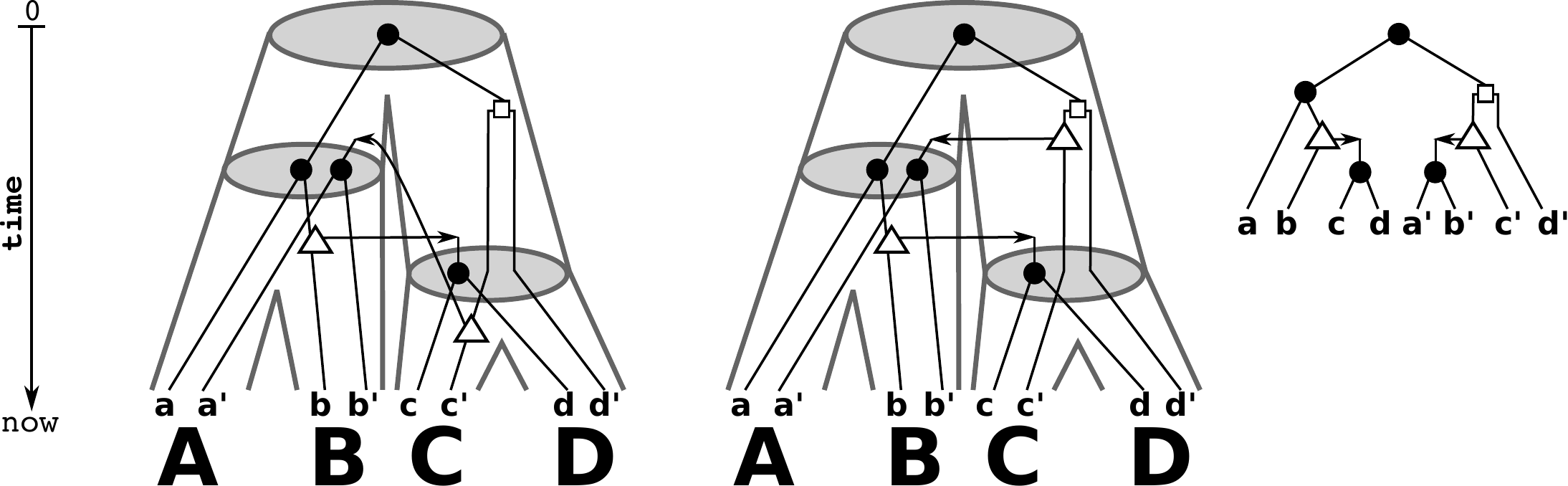}
  \end{center}
  \caption{Shown are a gene tree $(T;t,\sigma)$  (right)
    and two identical  (tube-like)  species trees $S$ (left and middle).
    There are two possible reconciliation maps for $T$ and $S$ that are
    given implicitly by drawing $T$ within the species tree $S$. These two
    reconciliation maps differ only in the choice of placing the HGT-event
    either on the edge $(\lca_S(C,D), C)$ or on the edge
    $(\lca_S(\{A,B,C,D\}), \lca_S(C,D))$. In the first case, it is easy to
    see that $\mu$ would not be time-consistent, i.e., there are no time
    maps $\tT$ and $\ts$ that satisfy (C1) and (C2).  The reconciliation
    map $\mu$ shown in the middle is time-consistent.  }
  \label{fig:nonU-Mu}
\end{figure}

\begin{definition}[Time Map]\label{def:time-map}
  The map $\tT: V(T) \rightarrow \mathbb{R}$ is a time map for the 
  rooted tree $T$ if $x\prec_T y$ implies $\tT(x)>\tT(y)$ for all 
  $x,y\in V(T)$. 
\end{definition}

\begin{definition} \label{def:tc-mu} A reconciliation map $\mu$ from
  $(T;t,\sigma)$ to $S$ is \emph{time-consistent} if there are time maps
  $\tau_T$ for $T$ and $\tau_S$ for $S$ for all $u\in V(T)$ satisfying the
  following conditions:
  \begin{description}
  \item[(C1)] If $t(u) \in \{\bullet, \odot \}$, then 
    $\tau_T(u) = \tau_S(\mu(u))$. \label{bio1}
  \item[(C2)] If $t(u)\in \{\square,\triangle \}$ and, thus
    $\mu(u)=(x,y)\in E(S)$, \label{bio2} then
    $\tau_S(y)>\tau_T(u)>\tau_S(x)$. 
\end{description}
\end{definition} 

Condition (C1) is used to identify the time-points of speciation vertices 
and leaves $u$ in the gene tree with the time-points of their respective 
images $\mu(u)$ in the species trees. In particular, all genes $u$ that
reside in the same species must be assigned the same time point
$\tau_T(u)=\tau_S(\sigma(u))$. Analogously, all speciation vertices in 
$T$ that are mapped to the same speciation in $S$ are assigned matching
time stamps, i.e., if $t(u)=t(v)=\bullet$ and $\mu(u)=\mu(v)$ then 
$\tau_T(u)=\tau_T(v)=\tau_S(\mu(u))$. 

To understand the intuition behind (C2) consider a duplication or HGT
vertex $u$. By construction of $\mu$ it is mapped to an edge of $S$, i.e.,
$\mu(u)=(x,y)$ in $S$. The time point of $u$ must thus lie between time
points of $x$ and $y$. Now suppose $(u,v)\in \EHGT$ is a transfer edge. By
construction, $u$ signifies the transfer event itself. The node $v$,
however, refers to the next (visible) event in the gene tree. Thus
$\tau_T(u)<\tau_T(v)$. In particular, $\tau_T(v)$ must not be
misinterpreted as the time of introducing the HGT-duplicate into the 
new lineage. While this time of course exists (and in our model coincides
with the timing of the transfer event) it is not marked by a visible event
in the new lineage, and hence there is no corresponding node in the gene
tree $T$. 

W.l.o.g.\ we fix the time axis so that $\tT(\rho_T) = 0$ and
$\ts(\rho_S) = -1$.  Thus, $\ts(\rho_S)< \tT(\rho_T) < \tT(u)$ for
all $u\in V(T)\setminus\{\rho_T\}$. 


Clearly, a necessary condition to have biologically feasible gene trees is
the existence of a reconciliation map $\mu$. 
However, not all
reconciliation maps are time-consistent, see Fig.~\ref{fig:Mu}. 

\begin{definition}
  An event-labeled gene tree $(T;t,\sigma)$ is \emph{biologically feasible} if
  there exists a time-consistent reconciliation map from $(T;t,\sigma)$ to
  some species tree $S$.
\end{definition}                

As a main result of this contribution, we provide simple conditions that
characterize (the existence of) time-consistent reconciliation maps and
thus, provides a first step towards the characterization of biologically
feasible gene trees.

\begin{figure}[tbp]
  \begin{center}
    \includegraphics[width=0.85\textwidth]{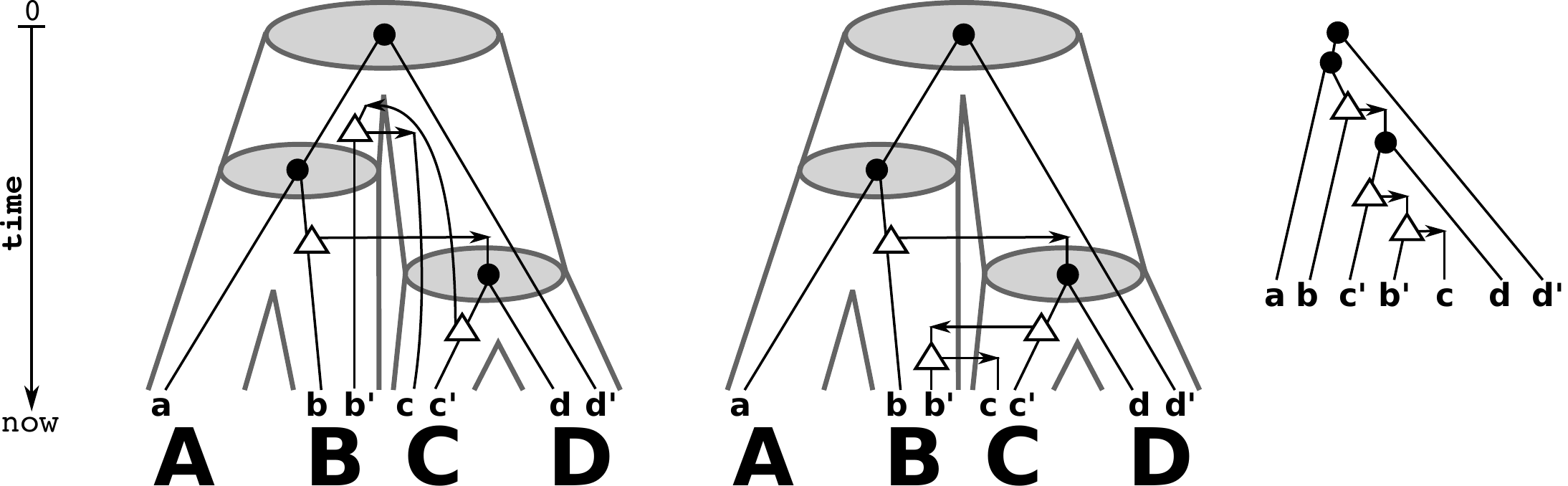}
  \end{center}
  \caption{Shown are a gene tree $(T;t,\sigma)$ (right) and two identical
    (tube-like) species trees $S$ (left and middle).  There are two
    possible reconciliation maps for $T$ and $S$ that are given implicitly
    by drawing $T$ within the species tree $S$.  The left reconciliation
    maps each gene tree vertex as high as possible into the species tree.
    However, in this case only the middle reconciliation map is
    time-consistent.  }
  \label{fig:nonU-Mu2}
\end{figure}
 
\begin{theorem} \label{thm:D} Let $\mu$ be a reconciliation map from
  $(T;t,\sigma)$ to $S$.  There is a \emph{time-consistent} reconciliation
  map from $(T;t,\sigma)$ to $S$ if and only if there are two time-maps
  $\tT$ and $\ts$ for $T$ and $S$, respectively, such that the following
  conditions are satisfied for all $x\in V(S)$:
  \begin{description}
  \item[(D1)] If $\mu(u) = x$, for some $u\in V(T)$ then $\tau_T(u) =
    \tau_S(x)$.
  \item[(D2)] If $x \preceq_S \lca_S(\sT(u))$ for some $u \in V(T)$ with
    $t(u)\in \{\dpl, \hgt\}$, then $\tau_S(x) > \tau_T(u)$.
  \item[(D3)] If $\lca_S(\sT(u)\cup\sT(v)) \preceq_S x$ for some $(u, v)
    \in \mc{E}$, then $\tau_T(u) > \tau_S(x)$.
  \end{description}
\end{theorem}
\begin{proof}
  See Appendix \ref{sec:thm:D}.
\end{proof}

Interestingly, the existence of a time-consistent
reconciliation map from a gene tree $T$ to a species tree $S$ can be
characterized in terms of a time map defined on $T$, only (see Theorem 
\ref{thm:onlyTT} in Appendix \ref{sec:thm:onlyTT}).

From the algorithmic point of view it is desirable to design methods that
allow to check whether a reconciliation map is time-consistent.  Moreover,
given a gene tree $T$ and species tree $S$ we wish to decide whether there
exists a time-consistent reconciliation map $\mu$, and if so, we should be
able to construct $\mu$.

To this end, observe that any constraints given by Definition
\ref{def:time-map}, Theorem \ref{thm:D} (D2)-(D3), and Definition
\ref{def:tc-mu} (C2) can be expressed as a total order on $V(S)\cup V(T)$,
while the constraints (C1) and (D1) together suggest that we can treat the
preimage of any vertex in the species tree as a ``single vertex''. In fact
we can create an auxiliary graph in order to answer questions that are
concerned with time-consistent reconciliation maps.

\begin{definition}\label{def:auxG}
  Let $\mu$ be a reconciliation map from $(T;t,\sigma)$ to $S$.  The
  \emph{auxiliary graph} $A$ is defined as a directed graph with a vertex
  set $V(A) = V(S) \cup V(T)$ and an edge-set $E(A)$ that is
    constructed as follows
  \begin{description}
  \item[(A1)] For each $(u,v)\in E(T)$ we have $(u',v') \in E(A)$, where
    \begin{equation*}
      u' =
      \begin{cases}
        \mu(u) & \text{if } t(u) \in \{\lea, \spe \} \\
        u & \text{otherwise}
      \end{cases}
      ,\ v' =
      \begin{cases}
        \mu(v) & \text{if } t(v)\in \{ \lea, \spe \} \\
        v & \text{otherwise}
      \end{cases},
    \end{equation*}
  \item[(A2)] For each $(x,y)\in E(S)$ we have $(x,y) \in E(A)$.
  \item[(A3)] For each $u \in V(T)$ with $t(u) \in \{ \dpl, \hgt \}$ we
    have $(u, lca_S(\sT(u))) \in E(A)$.
  \item[(A4)] For each $(u,v) \in \mc{E}$ we have $(lca_S(\sT(u)\cup
    \sT(v)),u) \in E(A)$
  \item[(A5)] For each $u\in V(T)$ with $t(u) \in \{ \hgt, \dpl \}$ and
    $\mu(u) = (x,y) \in E(S)$ we have $(x,u)\in E(A)$ and $(u,y)\in E(A)$.
  \end{description}
  We define $A_1$ and $A_2$ as the subgraphs of $A$ that contain only
    the edges defined by (A1), (A2), (A5) and (A1), (A2), (A3), (A4),
    respectively.
\end{definition}
We note that the edge sets defined by conditions (A1) through (A5) are not
necessarily disjoint. The mapping of vertices in $T$ to edges in $S$ is
considered only in condition (A5). The following two theorems proved in
Appendices \ref{sec:thm:aux-C} and \ref{sec:thm:aux} are the key results of
this contribution.

\begin{algorithm}[t]
  \caption{Check if there is a time-consistent reconciliation map
  } \label{alg:all}
  \begin{algorithmic}[1]
    \Require{$S=(W,F)$ is a species tree for $T=(V,E)$. }
    \Let{$\ell$}{$\texttt{ComputeLcaSigma}((T;t,\sigma),
      S)$} \label{alg2:lca} \State{$\mu(u) \gets \emptyset$ for all $u\in
      V$} \Comment{``$\emptyset$'' means uninitialized} \label{alg2:mu}
    \For{all $u\in V$} \label{alg:all:for1} \If{$t(u) \in \{\spe, \lea \}$}
    $\mu(u) \gets \ell(u)$ 
    \Else $\ \mu(u)\gets (p(\ell(u)), \ell(u))$ \Comment{$p(\ell(u))$
      denotes the parent of $\ell(u)$} \label{alg:all:edge}
    \EndIf
    \EndFor \label{alg:all:for2}
    \State Compute the auxiliary graph $A_2$ \label{alg:all:A} \If{$A_2$
      contains a cycle} \label{alg:all:cycle} \Return \textit{``No
      time-consistent reconciliation map exists.''}
    \EndIf
    \State{Let $\tau: V(A_2) \rightarrow \mb{R}$ such that if $(x,y) \in
      E(A_2)$ then $\tau(x) < \tau(y)$} \label{alg:all:constt1} \State
    \Comment{W.l.o.g.\ we can assume that $\tau(x) \neq \tau(y)$ for all
      $x,y\in V(A_2)$} \label{alg:all:comment} \Let{$\tau_S$}{A time map
      such that $\tau_S(x) = \tau(x)$ for all $x \in
      W$} \label{alg:all:constt2} \Let{$\tau_T$}{A time map such that
      $\tau_T(u) = \tau(\mu(u))$ if $t(u) \label{alg:all:constt3} \in \{
      \spe, \lea \}$, otherwise $\tau_T(u) = \tau(u)$ for all $u \in
      V$.} \label{alg:all:constt4}
    \For{$u\in V$ where $t(u) \in \{ \dpl, \hgt \}$} \label{alg:all:forx}
    \While{it does not hold that $\tau_S(x) < \tau_T(u) < \tau_S(y)$ for
      $(x,y) = \mu(u)$} \Let{$\mu(u)$}{$(p(x),x)$} \label{alg:all:mu}
    \EndWhile
    \EndFor
    \State \Return $\mu$
  \end{algorithmic}
\end{algorithm}

\begin{theorem}\label{thm:aux-C}
  Let $\mu$ be a reconciliation map from $(T;t,\sigma)$ to $S$.  The map
  $\mu$ is time-consistent if and only if the auxiliary graph $A_1$ is a
  directed acyclic graph (DAG).
	\label{thm:dag125}
\end{theorem}

\begin{theorem}\label{thm:aux}
	Assume there is a reconciliation map from $(T;t,\sigma)$ to $S$. 
	There is a time-consistent reconciliation map from $(T;t, \sigma)$ to $S$ if and only 
	if the auxiliary graph $A_2$ is a DAG.
\end{theorem}

Naturally, Theorems \ref{thm:dag125} or \ref{thm:aux} can be used to devise
algorithms for deciding time-consistency. To this end, the efficient
computation of $\lca_S(\sT(u))$ for all $u\in V(T)$ is necessary. This can
be achieved with Algorithm \ref{alg:compute-sigma} (Appendix
\ref{proof-Algos}) in $O(|V(T)|\log(|V(S)|))$. More precisely, we have the
following statement, which is proved in Appendix \ref{proof-Algos}.
\begin{lemma}
  For a given gene tree $(T=(V,E);t,\sigma)$ and a species tree $S=(W,F)$,
  Algorithm \ref{alg:compute-sigma} correctly computes $\ell(u) =
  \lca_S(\sT(u))$ for all $u \in V(T)$ in $O(|V(T)|\log(|V(S)|))$ time.
\label{lem:lca}
\end{lemma}

Let $S$ be a species tree for $(T;t,\sigma)$, that is, there is a valid
reconciliation between the two trees.  Algorithm \ref{alg:all} describes a
method to construct a time-consistent reconciliation map for $(T;t,\sigma)$
and $S$, if one exists, else ``No time-consistent reconciliation map
exists'' is returned.  First, an arbitrary reconciliation map $\mu$ that
satisfies the condition of Def.\ \ref{def:mu} is computed.  Second, Theorem
\ref{thm:aux} is utilized and it is checked whether the auxiliary graph
$A_2$ is not a DAG in which case no time-consistent map $\mu$ exists for
$(T;t,\sigma)$ and $S$. Finally, if $A_2$ is a DAG, then we continue to
adjust $\mu$ to become time-consistent. The latter is based on Thm.\
$\ref{thm:D}$, see the proof of Thm.\ $\ref{thm:D}$ and \ref{thm:main-algo} 
for details. In Appendix \ref{proof-Algos} we finally prove
\begin{theorem}
  Let $S= (W,F)$ be species tree for the gene tree $T=(V,E)$. Algorithm
  \ref{alg:all} correctly determines whether there is a time-consistent
  reconciliation map $\mu$ and in the positive case, returns such a $\mu$ in
  $O(|V|\log(|W|))$ time.
  \label{thm:main-algo}
\end{theorem}

\section{Outlook and Summary}

We have characterized here whether a given event-labeled gene tree $(T;
t,\sigma)$ and species tree $S$ can be reconciled in a time-consistent
manner in terms of two auxiliary graphs $A_1$ and $A_2$ that must be DAGs.
These are defined in terms of given reconciliation maps. This condition
yields an $O(|V|\log(|W|))$-time algorithm to check whether a given
reconciliation map $\mu$ is time-consistent, and an algorithm with the same
time complexity for the construction of a time-consistent reconciliation
maps, provided one exists.  The algorithms are implemented in \texttt{C++}
using the boost graph library and are freely available at
\url{https://github.com/Nojgaard/tc-recon}.

Our results depend on three conditions on the event-labeled gene trees
  that are motivated by the fact that event-labels can be assigned to
  internal vertices of gene trees only if there is observable information on
  the event. The question which event-labeled gene trees are actually
  observable given an arbitrary, true evolutionary scenario deserves further
  investigation in future work. Here we have used conditions that arguable
  are satisfied when gene trees are inferred using sequence comparison and
  synteny information. A more formal theory of observability is still
  missing, however.

Our results provide an efficient way of deciding whether a \emph{given}
pair of gene and species tree can be time-consistently reconciled. There
are, however, in general exponentially many putative species trees. This
begs the question whether there is \emph{at least one} species tree $S$ for
a gene tree and if so, how to construct $S$.  ``Informative triples''
  extracted from the gene tree answer this question in the absence of HGT
  \cite{HHH+12}. It is plausible that this idea can be
  generalized to our current setting to provide at least a partial
  characterization \cite{Hellmuth:17}.

\subparagraph*{Acknowledgment} We thank the organizers of the 32nd TBI
Winterseminar 2017 in Bled (Slovenia), where the authors participated, met
and basically drafted the main ideas of this paper, while drinking a cold
and tasty red Union, or was it a green La{\v{s}}ko?

\bibliographystyle{plain}
\bibliography{biblio}


\clearpage
\appendix
\noindent {\bf \Huge Appendix}
\\[0.2cm]

\section{Observable Scenarios}
\label{sec:app:obs}

\par\noindent\textbf{Lemma~\ref{lem:part_gen}.} 
Condition (O2) implies that 
$\{L_{\Th}(\rho_1), \dots, L_{\Th}(\rho_k)\}$ forms a partition of $\Gen$.
\smallskip
\begin{proof}
  Since $L_{\Th}(\rho_i)\subseteq V(T)$, it suffices to show that
  $L_{\Th}(\rho_i)$ does not contain vertices of $V(T)\setminus \Gen$.
  Note, $x\in L_{\Th}(\rho_i)$ with $x\notin\Gen$ is only possible if all
  edges $(x,y)$ are removed.
	
  Let $x\in V$ with $t(x) = \triangle$ such that all edges $(x,y)$ are
  removed.  Thus, all such edges $(x,y)$ are contained in $\EHGT$.
  Therefore, there every edge of the form $(x,y)$ is a transfer edge; a
  contradiction to (O2).
\end{proof}
\medskip 

Recall that all edges labeled ``$0$'' transmit the genetic material
vertically, i.e., from one species a descendant lineage. \\

\par\noindent\textbf{Lemma.} Conditions (O1) -- (O3) imply ($\Sigma$1).
\smallskip
\begin{proof}
Since (O2) is satisfied we can apply Lemma \ref{lem:part_gen} and conclude that neither 
$\sigma(L_{\Th}(v))=\emptyset$ nor  $\sigma(L_{\Th}(w))=\emptyset$. 
Let $x \in V(T)$ with $t(x)=\bullet$. 
By Condition (O1)
$x$ has (at least two) children. 
Moreover, (O3) 
implies that there are (at least) two children $v$ and $w$ in $T$ that are contained 
in distinct species $V$ and $W$ that are incomparable in $S$. 
Note, the edges $(x,v)$ and $(x,w)$ remain in $\Th$, since only
transfer edges are removed. 
Since no transfer is contained in $\Th$, the genetic material $v$ and $w$
of $V$ and $W$, respectively, is always vertically transmitted. 
Therefore, for any leaf $v'\in L_{\Th}(v)$ we have $\sigma(v')\preceq_S V$ 
and for any leaf $w'\in L_{\Th}(w)$ we have $\sigma(w')\preceq_S W$ in $S$. 
Assume now for contradiction, that 
 $\sigma(L_{\Th}(v))\cap \sigma(L_{\Th}(w)) \neq \emptyset$. 
Let $z_1\in L_{\Th}(v)$ and $z_2\in  L_{\Th}(w)$
with $\sigma(z_1) = \sigma(z_2) = Z$. 
Since $Z\preceq_S V,W$ and $S$ is a tree, the species $V$ and $W$
must be comparable in $S$; a contradiction to (O3).
\end{proof}

\par\noindent\textbf{Lemma.} Conditions (O1) -- (O3) imply ($\Sigma$2).
\smallskip
\begin{proof}
Since (O2) is satisfied we can apply Lemma \ref{lem:part_gen} and conclude that neither 
  $\sigma(L_{\Th}(v))=\emptyset$ nor $\sigma(L_{\Th}(w))=\emptyset$.  
	Let
  $(v,w) \in \EHGT$. 
 
	By (O3) the species containing $V$ and $W$ are are incomparable in $S$. 
	Now
  we can argue along the same lines as in the proof of the previous Lemma
  to conclude that $\sigma(L_{\Th}(v))\cap \sigma(L_{\Th}(w)) = \emptyset$.
\end{proof}

\section{DTL-scenario}
\label{sec:app:dtl}

In case that the event-labeling of $T$ is unknown, but the gene tree $T$
and a species tree $S$ are given, the authors in \cite{THL:11,BAK:12}
provide an axiom set, called DTL-scenario, to reconcile $T$ with $S$. This
reconciliation is then used to infer the event-labeling $t$ of $T$. Instead
of defining a DTL-scenario as octuple \cite{THL:11,BAK:12}, we use the
notation established above:
\begin{definition}[DTL-scenario]
  For a given gene tree $(T;t,\sigma)$ on $\Gen$ and a species tree $S$ on
  $\Spe$ the map $\gamma:V(T)\to V(S)$ maps the gene tree into the species
  tree such that
  \begin{description}
  \item[(I)] For each leaf $x\in \Gen$, $\gamma(u) = \sigma(u)$.
  \item[(II)] If $u\in V(T)\setminus\Gen$ with children $v,w$, then
    \begin{itemize}
    \item[(a)] $\gamma(u)$ is not a proper descendant of $\gamma(v)$ or
      $\gamma(w)$, and
    \item[(b)] at least one of $\gamma(v)$ or $\gamma(w)$ is a descendant
      of $\gamma(u)$.
    \end{itemize}
  \item[(III)] $(u,v)$ is a transfer edge if and only if $\gamma(u)$ and
    $\gamma(v)$ are incomparable.
  \item[(IV)] If $u\in V(T)\setminus\Gen$ with children $v,w$, then
    \begin{itemize}
    \item[(a)] $t(u)=\triangle$ if and only if either $(u,v)$ or $(u,w)$ is
      a transfer-edge,
    \item[(b)] If $t(u)=\bullet$, then $\gamma(u) =
      \lca_S(\gamma(v),\gamma(w))$ and $\gamma(v),\gamma(w)$ are
      incomparable,
    \item[(c)] If $t(u)=\square$, then
      $\gamma(u)\succeq\lca_S(\gamma(v),\gamma(w))$.
    \end{itemize}
  \end{description}
  \label{def:dtl}
\end{definition}

DTL-scenarios are explicitly defined for fully resolved binary gene and
species trees.  Indeed, Fig. \ref{fig:counterDTL} (right) shows a valid
reconciliation between a gene tree $T$ and a species tree $S$ that is not
consistent with DTL-scenario. To see this, let us call the duplication
vertex $v$. The vertex $v$ and the leaf $a$ are both children of the
speciation vertex $\rho_T$.  Condition (IVb) implies that $a$ and $v$ must
be incomparable. However, this is not possible since $\gamma(v)\succeq_S
\lca_S(B,C)$ (Cond.\ (IVc)) and $\gamma(a)=A$ (Cond.\ (I)) and therefore,
$\gamma(v)\succeq_S \lca_S(B,C) = \lca_S(A,B,C) \succ_S \gamma(a)$.

Nevertheless, we show in the following that, in case both gene and species
trees are binary, our choice of reconciliation map is equivalent to the
definition of a DTL-scenario \cite{THL:11,BAK:12}.  To this end, we provide
first the following lemmas that establishes useful properties of the
reconciliation map

\sloppy
\begin{lemma}
	Let $\mu$ be a reconciliation map from $(T;t,\sigma)$ to $S$ and assume that $T$ is binary.
	 Then the following conditions are satisfied:
	\begin{enumerate}
		\item 	If $v,w\in V(T)$ are in the same connected component of $\Th$, then \\
					$\mu(\lca_{\Th}(v,w)) \succeq_S \lca_S(\mu(v),\mu(w))$.
			\label{item:lca}
	\end{enumerate}
	Let	$u$ be an arbitrary interior vertex of $T$ with children $v,w$, then:
	\begin{enumerate}  \setcounter{enumi}{1}
		\item 	$\mu(u)$ and $\mu(v)$ are incomparable in $S$ if and only if $(u,v)\in \EHGT$. \label{item:hgt}
		\item		If $t(u)=\bullet$, then	$\mu(v)$ and $\mu(w)$ are incomparable in $S$. \label{item:speci}
		\item 	If $\mu(v), \mu(w)$ are comparable or $\mu(u)\succ_S \lca_S(\mu(v),\mu(w))$, then $t(u)=\square$. \label{item:dupli}
	\end{enumerate}
\label{lem:cond-mu}
\end{lemma}
\begin{proof} We prove the Items \ref{item:lca} - \ref{item:dupli} separately. Recall, 
		Lemma \ref{lem:part_gen} 
		implies that 	$\sigma(L_{\Th}(x)) \neq \emptyset$ for all $x\in V(T)$.

		\emph{Proof of Item \ref{item:lca}:}
		Let $v$ and $w$ be distinct vertices of $T$ that are in the same connected
		component of $\Th$. Consider the unique path $P$ connecting $w$ with $v$
		in $\Th$. This path $P$ is uniquely subdivided into a path $P'$ and a path
		$P''$ from $\lca_{\Th}(v,w)$ to $v$ and $w$, respectively. Condition (M3)
		implies that the images of the vertices of $P'$ and $P''$ under $\mu$,
		resp., are ordered in $S$ with regards to $\preceq_S$ and hence, are
		contained in the intervals $Q'$ and $Q''$ that connect
		$\mu(\lca_{\Th}(v,w))$ with $\mu(v)$ and $\mu(w)$, respectively. In
		particular, $\mu(\lca_{\Th}(v,w))$ is the largest element (w.r.t.\
		$\preceq_S$) in the union of $Q'\cup Q''$ which contains the unique path
		from $\mu(v)$ to $\mu(w)$ and hence also $\lca_S(\mu(v),\mu(w))$.

		\emph{Proof of Item \ref{item:hgt}:}
		If $(u,v)\in \EHGT$ then, $t(u)=\triangle$ and (M2iii) implies that 
		$\mu(u)$ and $\mu(v)$ are incomparable.
		To see the converse, let $\mu(u)$ and $\mu(v)$
		be incomparable in $S$. Item (M3) implies that for any edge $(x,y)\in
		E(\Th)$ we have $\mu(y)\preceq_S \mu(x)$. However, since $\mu(u)$ and
		$\mu(v)$ are incomparable it must hold that $(u,v)\notin E(\Th)$. Since
		$(u,v)$ is an edge in the gene tree $T$, $(u,v)\in \EHGT$ is a transfer
		edge.

		\emph{Proof of Item \ref{item:speci}:} 
		Let $t(u)=\bullet$. 
		Since none of $(u,v)$ and $(u,w)$ are transfer-edges, it follows that 
		both edges are contained in $\Th$. 
		Then, since $T$ is a binary tree, it follows that $L_{\Th}(u) = L_{\Th}(v)\cup L_{\Th}(w)$ 
		and therefore, $\sT(u)  = \sT(v)\cup \sT(w)$.

		Therefore and by Item (M2i), 
		\[	\mu(u) = \lca_S(\sT(u))   =  	\lca_S(\sT(v)\cup \sT(w))
						  = \lca_S(\lca_S(\sT(v)), \lca_S(\sT(w))). \]
		
		Assume for contradiction that $\mu(v)$ and $\mu(w)$ are comparable, say,
		$\mu(w) \succeq_S \mu(v)$. By Lemma \ref{lem:cond-mu2}, $\mu(w) \succeq_S \mu(v) \succeq_S
		\lca_S(\sT(v))$ and $\mu(w) \succeq_S
		\lca_S(\sT(w))$. Thus, 
		\[\mu(w) \succeq_S \lca_S(\lca_S(\sT(v)), \lca_S(\sT(w))). \]
		Thus, $\mu(w) \succeq_S \mu(u) $; a contradiction to (M3ii).

      \emph{Proof of Item \ref{item:dupli}:} 
		Let $\mu(v), \mu(w)$ be comparable in $S$. Item \ref{item:speci} implies
		that $t(u)\neq \bullet$. Assume for contradiction that $t(u) = \triangle$.
		Since by (O2) only one of the edges $(u,v)$ and $(u,w)$ is a transfer edge, we
		have either $(u,v)\in \EHGT$ or $(u,w)\in \EHGT$. W.l.o.g. let $(u,v)\in
		\EHGT$ and $(u,w)\in E(\Th)$. By Condition (M3), $\mu(u)\succeq_S\mu(w)$.
		However, since $\mu(v)$ and $\mu(w)$ are comparable in $S$, also $\mu(u)$
		and $\mu(v)$ are comparable in $S$; a contradiction to Item
		\ref{item:hgt}. Thus, $t(u)\neq \triangle$. Since each interior vertex is
		labeled with one event, we have $t(u) = \square$. 

		Assume now that $\mu(u)\succ_S \lca_S(\mu(v),\mu(w))$. Hence, $\mu(u)$ is
		comparable to both $\mu(v)$ and $\mu(w)$ and thus, (M2iii) implies that
		$t(u) \neq \triangle$.  
		Lemma \ref{lem:cond-mu2} implies $\mu(v)\succeq_S \lca_S(\sT(v))$ 
 		and $\mu(w)\succeq_S \lca_S(\sT(w))$.
		\begin{equation*} 
		\lca_S(\mu(v),\mu(w)) \succeq_S  \lca_S(\lca_S(\sT(v)), \lca_S(\sT(w))) 
		= \lca_S(\sT(v)\cup \sT(w)).
		\end{equation*}
		Since $T(u) \neq \triangle$ it follows that neither $(u,v)\in \EHGT$ nor $(u,w)\in \EHGT$
		and hence, both edges are contained in $\Th$. 
		By the same argumentation as in Item \ref{item:speci} it follows that
		$\sT(u) = \sT(v)\cup \sT(w)$ and therefore, $\lca_S(\sT(v)\cup \sT(w)) = \lca_S(\sT(u))$.
		Hence, $\mu(u) \succ_S \lca_S(\mu(v),\mu(w)) \succeq_S 
		\lca_S(\sT(u))$. Now, (M2i) implies $t(u)\neq \bullet$. Since
		each interior vertex is labeled with one event, we have $t(u) = \square$. 
\hfill \end{proof}

\begin{lemma}
	Let $\mu$ be a reconciliation map for the gene tree  $(T;t,\sigma)$ and the species tree $S$
	as in Definition \ref{def:mu}. Moreover, assume that $T$ and $S$ are binary.
	Set for all $u\in V(T)$:
	\begin{equation*}
  		\gamma(u) = \begin{cases}
					        \mu(u) &\mbox{,if } \mu(u)\in V(S) \\
       					 		y 	&\mbox{,if } \mu(u) = (x,y)\in E(S)
        \end{cases}
 	\end{equation*}
	Then  $\gamma:V(T)\to V(S)$ is a map according to the DTL-scenario.
	\label{lem:mu-dtl}
\end{lemma}
\begin{proof}
	We first emphasize that, by construction, $\mu(u)\succeq_S \gamma(u)$ for all
	$u\in V(T)$. Moreover, $\mu(u) =\mu(v)$ implies that $\gamma(u) =\gamma(v)$,
	and $\gamma(u) =\gamma(v)$ implies that $\mu(u)$ and $\mu(v)$ are comparable.
	Furthermore, $\mu(u)\prec_S \mu(v)$ implies $\gamma(u)\preceq_S \gamma(v)$,
	while $\gamma(u)\prec_S \gamma(v)$ implies that $\mu(u)\prec_S \mu(v)$.
	Thus, $\mu(u)$ and $\mu(v)$ are comparable if and only if $\gamma(u)$ and $\gamma(v)$
	are comparable.
	
	Item (I) and (M1) are equivalent. 
	
	For Item (II) let $u\in V(T)\setminus \Gen$ be an interior vertex with
	children $v,w$. If $(u,w) \notin \EHGT$, then $w\prec_{\Th} u$. Applying
	Condition (M3) yields $\mu(w)\preceq_{S} \mu(u)$ and thus, by construction,
	$\gamma(w)\preceq_{S} \gamma(u)$. Therefore, $\gamma(u)$ is not a proper
	descendant of $\gamma(w)$ and $\gamma(w)$ is a descendant of $\gamma(u)$. If
	one of the edges, say $(u,v)$, is a transfer edge, then $t(u) = \triangle$
	and by Condition (M2iii) $\mu(u)$ and $\mu(v)$ are incomparable. Hence,
	$\gamma(u)$ and $\gamma(v)$ are incomparable. Therefore, $\gamma(u)$ is no
	proper descendant of $\gamma(v)$. Note that (O2) implies that for each vertex $u\in
	V(T)\setminus \Gen$ at least one of its outgoing edges must be a non-transfer
	edge, which implies that $\gamma(w)\preceq_{S} \gamma(u)$ or
	$\gamma(v)\preceq_{S} \gamma(u)$ as shown before. Hence, Item (IIa) and (IIb)
	are satisfied.

	For Item (III) assume first that $(u,v) \in \EHGT$ and therefore $t(u) = \triangle$.
	Then, (M2iii) implies that $\mu(u)$ and $\mu(v)$ are incomparable and thus,
	$\gamma(u)$ and $\gamma(v)$ are incomparable.
	Now assume that $(u,v)$ is an edge in the gene tree $T$ and
	$\gamma(u)$ and $\gamma(v)$ are incomparable. Therefore, $\mu(u)$ and
	$\mu(v)$ are incomparable. Now, apply Lemma
	\ref{lem:cond-mu}(\ref{item:hgt}).

	Item (IVa) is clear by the event-labeling $t$ of $T$ and since (O2). Now assume for (IVb)
	that $t(u) = \bullet$. Lemma \ref{lem:cond-mu}(\ref{item:speci}) implies that
	$\mu(v)$ and $\mu(w)$ are incomparable and thus, $\gamma(v)$ and $\gamma(w)$
	must be incomparable as well. Furthermore, Condition (M2i) implies that
	$\mu(u) = \lca_S(\sT(u))$. Lemma
	\ref{lem:cond-mu2} implies that $\mu(v) \succeq_S
	\lca_S(\sT(v))$ and $\mu(w) \succeq_S
	\lca_S(\sT(w))$. The latter together with the incomparability of
	$\mu(v)$ and $\mu(u)$ implies that 
	\begin{equation*} 
			\begin{split}
				\lca_S(\mu(v),\mu(w)) & = \lca_S(\lca_S(\sT(v)), \lca_S(\sT(w))) \\
				&	= \lca_S(\sT(v)\cup \sT(w))=\lca_S(\sT(u)) =\mu(u).
			\end{split}
		\end{equation*}
	If $\mu(v)$ is mapped on the edge $(x,y)$ in $T$, then $\gamma(v) = y$. By
	definition of $\lca$ for edges, $\lca_S(\mu(v),\gamma(w)) =
	\lca_S(y,\gamma(w)) = \lca_S(\gamma(v),\gamma(w))$. The same argument applies
	if $\mu(w)$ is mapped on an edge. Since for all $z\in V(T)$ either
	$\mu(z)\succ_S \gamma(z)$ (if $\mu(z)$ is mapped on an edge) or $\mu(z)=
	\gamma(z)$, we always have
	\[\lca_S(\gamma(v),\gamma(w)) = \lca_S(\mu(v),\mu(w)) = \mu(u) . \]

	Since $t(u)=\bullet$, (M2i) implies that $\mu(u) \in V(S)$ and therefore, by construction of $\gamma$
	it holds that $\mu(u)=\gamma(u)$. Thus, $\gamma(u) = \lca_S(\gamma(v),\gamma(w))$.
	For (IVc) assume that $t(u) = \square$.
	Condition (M3) implies that $\mu(u)\succeq_S \mu(v), \mu(w)$ and
	therefore, $\gamma(u)\succeq_S \gamma(v), \gamma(w)$. If $\gamma(v)$ and
	$\gamma(w)$ are incomparable, then $\gamma(u)\succeq_S \gamma(v), \gamma(w)$
	implies that $\gamma(u)\succeq_S \lca_S(\gamma(v), \gamma(w))$. If
	$\gamma(v)$ and $\gamma(w)$ are comparable, say $\gamma(v)\succeq_S
	\gamma(w)$, then $\gamma(u)\succeq_S \gamma(v) = \lca_S(\gamma(v),
	\gamma(w))$. Hence, Statement (IVc) is satisfied. 
\hfill \end{proof}

\begin{lemma}
	Let $\gamma:V(T)\to V(S)$ be a map according to the DTL-scenario for the binary 
	the gene tree  $(T;t,\sigma)$ and the binary species tree $S$.
	
	Set for all $u\in V(T)$:
	\begin{equation*}
  		\mu(u) = \begin{cases}
					        \gamma(u) &\mbox{,if } t(u)\in \{\bullet, \odot\} \\
       					 	(x,\gamma(u)) \in E(S)	&\mbox{,if } t(u)\in \{\triangle, \square\}
        \end{cases}
 	\end{equation*}
	Then  $\mu:V(T)\to V(S)\cup E(S)$ is a reconciliation map according to Definition \ref{def:mu}. 
	\label{lem:dtl-mu}
\end{lemma}
\begin{proof}
	Let $\gamma:V(T)\to V(S)$ be a map a DTL-scenario for the binary 
	the gene tree  $(T;t,\sigma)$ and the species tree $S$. 
	
	Condition (M1) is equivalent to (I). 

	For (M3) assume that $v \preceq_{\Th} w$. The path $P$ from $v$ to $w$ in $\Th$	
	does not contain transfer edges. Thus, by (III) all vertices along $P$ are comparable. 
	Moreover, by (IIa) we have that $\gamma(w)$ is not a proper descendant of the image of its child in $S$, 
	and therefore, 
	by repeating these arguments along the vertices $x$ in $P_{wv}$, we obtain
	$\gamma(v)\preceq_S\gamma(x)\preceq_S\gamma(w)$. 
	
	If $\gamma(v)\prec_S\gamma(w)$, 
	then by construction of $\mu$, it follows that $\mu(v)\prec_S\mu(w)$.
	Thus, (M3) is satisfied, whenever $\gamma(v)\prec_S\gamma(w)$.
	Assume now that $\gamma(v)=\gamma(w)$.
	If $t(v),t(w) \in \{\square,\triangle\}$ then $\mu(v)=(x,\gamma(v))=(x,\gamma(w))=\mu(w)$
	and thus (M3i) is satisfied.
	If $t(v)=\bullet$ and $t(w)\neq\bullet$ then since $\mu(v)=\gamma(v)$ and $\mu(w) = (x,\gamma(w))$.
	Thus $\mu(v) \prec_S \mu(w)$.

	Now assume that    $\gamma(v)  =\gamma(w)$ and $w$ is a speciation
	vertex. Since $t(w) = \bullet$, for its two children $w'$ and $w''$ 
	the images $\gamma(w')$ and $\gamma(w'')$
	must be incomparable due to (IVb). W.l.o.g. assume that $w'$ is a vertex of $P_{wv}$.
	Since $\gamma(v)\preceq_S\gamma(x)\preceq_S\gamma(w)$
	for any vertex $x$ along $P_{wv}$ and $\gamma(v)  =\gamma(w)$, we obtain
	$\gamma(w')  =\gamma(w)$. However, since $\gamma(w'') \preceq_S \gamma(w)$, 
	the vertices $\gamma(w')$ and $\gamma(w'')$ are comparable in $S$;  contradicting (IVb).   
	Thus, whenever $w$ is a speciation vertex, $\gamma(w')  =\gamma(w)$  
	is not possible. 
	Therefore, $\gamma(v)\preceq_S \gamma(w')  \prec_S\gamma(w)$ and, 
	by construction of $\mu$, $\mu(v)  \prec_S\mu(w)$. Thus, (M3ii) is satisfied. 

	Finally, we show that (M2) is satisfied. To this end, observe first that
	(M2ii) is fulfilled by construction of $\mu$ and (M2iii) is an immediate 
	consequence of (III). Thus, it remains to show that (M2i) is satisfied. 
	Thus, for a given speciation vertex $u$ we need to show that 
	$\mu(u)=\lca_S(\sT(u))$. 
	By construction, $\mu(u) = \gamma(u)$. 
	Note, $\Th$ does not contain transfer edges. Applying (III) implies that for 
	all edges $(x,y)$ in $\Th$ the images $\gamma(x)$ and $\gamma(y)$ must be 
	comparable. The latter and  (IIa) implies that for all edges $(x,y)$ in $\Th$ 
	we have $\gamma(y) \preceq_S \gamma(x)$. Take the latter together,  
	$\sigma(z)=\gamma(z) \preceq_S\gamma(u)$ for any leaf $z\in L_{\Th}(u)$. 
	Therefore $\lca_S(\sT(u)) \preceq_S\gamma(u) = \mu(u)$. 
	Assume for contradiction that $\lca_S(\sT(u)) \prec_S\gamma(u) = \mu(u)$.
	Consider the two children $u'$ and $u''$ of $u$ in $\Th$. 
	Since neither $(u,u') \in \EHGT$ nor $(u,u'') \in \EHGT$ and $T$ is a binary tree, it follows that
	$L_{\Th}(u) = L_{\Th}(u')\cup L_{\Th}(u'')$ and we obtain that 
	$\sT(u)  = \sT(u') \cup \sT(u'')$. Moreover, re-using the arguments above, 
	$\lca_S(\sT(u')) \preceq_S\gamma(u')$ and $\lca_S(\sT(u'')) \preceq_S\gamma(u'')$. 
	By the arguments we used in the proof for (M3), we have 
	$\gamma(u')\prec_S \gamma(u)$	and $\gamma(u'')\prec_S \gamma(u)$. 
    In particular, $\gamma(u')$ and $\gamma(u'')$ must be contained in
	the subtree of $S$ that is rooted in the child $a$ of $\gamma(u)$
	in $S$ with $\lca_S(\sT(u))\preceq_S a$, as otherwise, 
	$\lca_S(\sT(u')) \not \preceq_S\gamma(u')$ or $\lca_S(\sT(u'')) \not\preceq_S\gamma(u'')$. 
	Moreover, neither $\lca_S(\sT(u))\preceq_S \lca_S(\sT(u'))$ nor 
	$\lca_S(\sT(u))\preceq_S \lca_S(\sT(u''))$ is possible since then 
	$\lca_S(\sT(u')) \preceq_S\gamma(u')$ and $\lca_S(\sT(u'')) \preceq_S\gamma(u'')$
	implies that $\gamma(u')$ and $\gamma(u'')$ would be comparable; contradicting   
	(IVb). Hence, 
	there remains only one way to locate $\gamma(u')$ and $\gamma(u'')$, 
	that is, they must be located in the subtree of $S$ that is rooted 
	in $\lca_S(\sT(u))$. But then we have 
	$\lca_S(\gamma(u'), \gamma(u''))\preceq_S \lca_S(\sT(u)) \prec_S \gamma(u)$;
	a contradiction to (IVb) $\gamma(u)  = \lca_S(\gamma(u'), \gamma(u''))$.
	Therefore,  $\lca_S(\sT(u))  = \gamma(u) = \mu(u)$ and (M2i) is satisfied. 
\end{proof}

Lemma \ref{lem:mu-dtl} and \ref{lem:dtl-mu} imply
\begin{theorem}
	For a binary gene tree $(T;t,\sigma)$ and a binary species tree $S$
	there is DTL-scenario if and only if there is a reconciliation $\mu$
	for $(T;t,\sigma)$ and $S$
\end{theorem}

\section{Proof of Theorem \ref{thm:D}}
\label{sec:thm:D}
\begin{proof} %
  In the following, $x$ and $u$ denote vertices in $S$ and $T$,
  respectively.

  $(\Longrightarrow)$ Assume that there is a time-consistent reconciliation
  map $\mu$ from $(T;t, \sigma)$ to $S$, and thus two time-maps $\tau_S$
  and $\tau_T$ for $S$ and $T$, respectively, that satisfy (C1) and (C2).

  To see (D1), observe that if $\mu(u) = x\in V(S)$, then (M1) and (M2)
  imply that $t(u) \in \{ \spe, \lea \}$.  Now apply (C1).

  To show (D2), assume that $t(u) \in \{ \dpl, \hgt \}$ and $x \preceq_S
  lca_S(\sT(u))$.  By Condition (M2) it holds that $\mu(u) = (y,z) \in
  E(S)$.  Together with Lemma \ref{lem:cond-mu2} we obtain that $x \preceq_S
  \lca_S(\sT(u)) \preceq_S z \prec_S \mu(u)$. By the properties of $\ts$ we
  have
  \begin{equation*}
    \ts(x) \geq \ts(\lca_S(\sT(u)) \geq \ts(z) \stackrel{(C2)}{>} \tT(u).
  \end{equation*}

  To see (D3), assume that $(u,v) \in \mc{E}$ and $z\coloneqq
  \lca_S(\sT(u)\cup \sT(v)) \preceq_S x$. Since $t(u)=\hgt$ and by (M2ii), we
  have $\mu(u) = (y,y')\in E(S)$. Thus, $\mu(u) \prec_S y$.  By (M2iii) $\mu(u)$
  and $\mu(v)$ are incomparable and therefore, we have either $\mu(v)\prec_S y$ or
  $\mu(v)$ and $y$ are incomparable.  In either case we see that $y
  \preceq_S z$, since Lemma \ref{lem:cond-mu2} implies that
  $\lca_S(\sT(u))\preceq_S\mu(u)$ and $\lca_S(\sT(v))\preceq_S\mu(v)$.  In
  summary, $\mu(u) \prec_S y \preceq_S z \preceq_S x$.  Therefore,
  \begin{equation*}
    \tT(u) \stackrel{(C2)}{>} \ts(y) \geq \ts(z) \geq \ts(x).
  \end{equation*}
  Hence, conditions (D1)-(D3) are satisfied.

  $(\Longleftarrow)$ To prove the converse, assume that there exists a
  reconciliation map $\mu$ that satisfies (D1)-(D3) for some time-maps
  $\tT$ and $\ts$. In the following we will make use of $\ts$ and $\tT$ 
  to construct a time-consistent reconciliation map $\mu'$. 

  First we define ``anchor points'' by $\mu'(v) =\mu(v)$ for all $v \in
  V(T)$ with $t(v) \in\{ \spe, \odot\}$. Condition (D1) implies $\tT(v) =
  \ts(\mu(v))$ for these vertices, and therefore $\mu'$ satisfies (C1).

  The next step will be to show that for each vertex $u\in V(T)$ with $t(u)
  \in \{ \dpl, \hgt \}$ there is a unique edge $(x,y)$ along the path from
  $\lca_S(\sT(u))$ to $\rho_S$ with $\ts(x)<\tT(u)<\ts(y)$. We set
  $\mu'(u) =(x,y)$ for these points.  In the final step we will show that
  $\mu'$ is a valid reconciliation map.
	
  Consider the unique path $\mc{P}_u$ from $\lca_S(\sT(u))$ to $\rho_S$.
  By construction, $\ts(\rho_S) <\tT(\rho_T) \leq \tT(u)$ and by Condition
  (D2) it we have $\tT(u) < \ts(\lca_S(\sT(u)))$.  Since $\ts$ is a time
  map for $S$, every edge $(x,y)\in E(S)$ satisfies
  $\ts(x)<\ts(y)$. Therefore, there is a unique edge $(x_u,y_u)\in E(S)$
  along $\mc{P}_u$ such that either $\ts(x_u)<\tT(u)<\ts(y_u)$,
  $\ts(x_u)=\tT(u)<\ts(y_u)$, or $\ts(x_u)<\tT(u)=\ts(y_u)$. The addition
  of a sufficiently small perturbation $\epsilon_u$ to $\tT(u)$ does not
  violate the conditions for $\tT$ being a time-map for $T$. Clearly
  $\epsilon_u$ can be chosen to break the equalities in the latter two cases
  in such a way that $\ts(x_u)<\tT(u)<\ts(y_u)$ for each vertex $u\in V(T)$
  with $t(u) \in \{ \dpl, \hgt \}$. We then continue with the perturbed
  version of $\tT$ and set $\mu'(u) =(x_u,y_u)$.  By construction, $\mu'$
  satisfies (C2).

  It remains to show that $\mu'$ is a valid reconciliation map from
  $(T;t,\sT)$ to $S$. Again, let $\mc{P}_u$ denote the unique path from
  $\lca_S(\sT(u))$ to $\rho_S$ for any $u\in V(T)$.

  By construction, Conditions (M1), (M2i), (M2ii) are satisfied.  To check
  condition (M2iii), assume $(u,v)\in \EHGT$.  The original map $\mu$ is a
  valid reconciliation map, and thus, Lemma \ref{lem:cond-mu2} implies that
  $\lca_S(\sT(u))\prec_S \mu(u)$ and $\lca_S(\sT(v))\preceq_S \mu(v)$.
  Since $\mu(u)$ and $\mu(v)$ are incomparable in $S$ and
  $\lca_S(\sT(u)\cup \sT(v))$ lies on both paths $\mc{P}_u$ and $\mc{P}_v$
  we have $\mu(u), \mu(v) \preceq_S \lca_S(\sT(u)\cup \sT(v))\eqqcolon x$.
  In particular, $x\neq \lca_S(\sT(u))$ and $x\neq\lca_S(\sT(v))$.
	
  Conditions (D1) and (D2) imply that $\ts(x)<\tT(u)<\ts(\lca_S(\sT(u)))$
  and $\ts(x)< \tT(v) \le\ts(\lca_S(\sT(v)))$.  By construction of $\mu'$, the
  vertex $u$ is mapped to a unique edge $e_u = (x_u,y_u)$ and $v$ is mapped
  either to $\lca_S(\sT(v))\neq x$ or to the unique edge $e_v=(x_v,y_v)$,
  respectively. In particular, $\mu'(u)$ lies on the path $\mc{P}'$ from
  $x$ to $\lca_S(\sT(u))$ and $\mu'(v)$ lies one the path $\mc{P}''$ from
  $x$ to $\lca_S(\sT(v))$. The paths $\mc{P}'$ and $\mc{P}''$ are
  edge-disjoint and have $x$ as their only common vertex.  Hence, $\mu'(u)$
  and $\mu'(v)$ are incomparable in $S$, and (M2iii) is satisfied.
	
  In order to show (M3), assume that $u\prec_{\Th} v$.  Since $u\prec_{\Th}
  v$, we have $\sT(u)\subseteq \sT(v)$.  Hence, $\lca_S(\sT(u)) \preceq
  \lca_S(\sT(v))\preceq_S\rho_S$.  In other words, $\lca_S(\sT(v))$ lies on
  the path $\mc{P}_u$ and thus, $\mc{P}_v$ is a subpath of $\mc{P}_u$.  By
  construction of $\mu'$, both $\mu'(u)$ and $\mu'(v)$ are comparable in
  $S$.  Moreover, since $\tT(u)>\tT(v)$ and by construction of $\mu'$, it
  immediately follows that $\mu'(u) \preceq_S\mu'(v)$.  

  Its now an easy task to verify that (M3) is fulfilled by considering the
  distinct event-labels in (M3i) and (M3ii), which we leave to the reader.
\end{proof}

\section{Theorem \ref{thm:onlyTT}}
\label{sec:thm:onlyTT}

\begin{theorem}
  Let $\mu$ be a reconciliation map from $(T;t, \sigma)$ to $S$. There is a
  time-consistent reconciliation map $(T;t, \sigma)$ to $S$ if and only if
  there is a time map $\tau_T$ such that for all $u,v,w\in V(T)$:
  \begin{description}
  \item[(T1)] If $t(u) = t(v) \in \{ \spe,\odot\}$ then
    \begin{enumerate}
    \item[(a)] If $\mu(u)=\mu(v)$, then $\tT(u) = \tT(v)$.
    \item[(b)] If $\mu(u)\prec_S \mu(v)$, then $\tT(u) > \tT(v)$.
    \end{enumerate}
  \item[(T2)] If $t(u)\in \{\spe, \odot\}$, $t(v)\in \{\dpl,\hgt\}$ and
    $\mu(u)\preceq_S \lca_S(\sT(v))$, then $\tau_T(u)>\tau_T(v)$.
  \item[(T3)] If $(u,v) \in \mc{E}$ and $\lca_S(\sT(u) \cup \sT(v))
    \preceq_S \lca_S(\sT(w))$ for some $w\in V(T)$, then $\tau_T(u) >
    \tau_T(w)$
  \end{description}
\label{thm:onlyTT}
\end{theorem}
\begin{proof} 
  Suppose that $\mu$ is a time-consistent reconciliation map from $(T;t,
  \sigma)$ to $S$.  By Definition \ref{def:tc-mu} and Theorem \ref{thm:D},
  there are two time maps $\tT$ and $\ts$ that satisfy (D1)-(D3).  We first
  show that $\tT$ also satisfies (T1)-(T3), for all $u,v\in
  V(T)$. Condition (T1a) is trivially implied by (D1).  Let $t(u),t(v)\in
  \{ \spe, \lea \}$, and $\mu(u) \prec_S \mu(v)$. Since $\tT$ and $\ts$ are
  time maps, we may conclude that 
  \begin{equation*}
    \tau_T(u) \stackrel{(D1)}{=} \tau_S(\mu(u)) < \tau_S(\mu(v))
    \stackrel{(D1)}{=} \tau_T(v).
  \end{equation*}
  Hence, (T1b) is satisfied.

  Now, assume that $t(u) \in \{ \spe, \lea \}$, $t(v) \in \{ \dpl, \hgt \}$
  and $\mu(u) \preceq_S \lca(\sT(v))$. By the properties of $\tau_S$, we
  have:
  \begin{equation*}
    \tau_T(u) \stackrel{(D1)}{=} \tau_S(\mu(u)) \stackrel{(D2)}{>} \tau_T(v).
  \end{equation*}
  Hence (T2) is fulfilled.

  Finally, assume that $(u,v) \in \mc{E}$, and $x \coloneqq \lca_S(\sT(u)
  \cup \sT(v)) \preceq_S \lca_S(\sT(w))$ for some $w \in V(T)$. Lemma
  \ref{lem:cond-mu2} implies that $\lca_S(\sT(w)) \preceq_S \mu(w)$ and we
  obtain
  \begin{equation*}
    \tau_T(w) \stackrel{(D2)}{<} \ts(x) \leq \tau_S(\lca(\sT(w))) 
    \stackrel{(D3)}{<} \tau_T(u).
  \end{equation*}
  Hence, (T3) is fulfilled.
	
  To see the converse, assume that there exists a reconciliation map $\mu$
  that satisfies (T1)-(T3) for some time map $\tT$. In the following we
  construct a time map $\ts$ for $S$ that satisfies (D1)-(D3).  To this
  end, we first set
  \begin{equation*}
    \ts(x) = \begin{cases} 
      -1     & \mbox{if }       x=\rho_S \\
      \tT(v) &\mbox{else if }   v\in \mu^{-1}(x)  \\
      \ast   &\mbox{else, i.e., } \mu^{-1}(x) =\emptyset 
               \text{ and } x\neq \rho_S.\,
             \end{cases} 
   \end{equation*}
   We use the symbol $\ast$ to denote the fact that so far no value has
   been assigned to $\ts(x)$.  Note, by (M2i) and (T1a) the value $\ts(x)$
   is uniquely determined and thus, by construction, (D1) is satisfied.
   Moreover, if $x,y\in V(S)$ have non-empty preimages w.r.t.\ $\mu$ and
   $x\prec_S y$, then we can use the fact that $\tT$ is a time map for $T$
   together with condition (T1) to conclude that $\ts(x)>\ts(y)$.

   If $x \in V(S)$ with $a\in \mu^{-1}(x)$, then (T2) implies (D2) (by (D1)
   and setting $u=a$ in (T2)) and (T3) implies (D3) (by (D1) and setting
   $w=a$ in (T3)).  Thus, (D2) and (D3) is satisfied for all $x \in V(S)$
   with $\mu^{-1}(x) \neq\emptyset$.  

   Using our choices $\ts(\rho_T) = 0$ and $\ts(\rho_S) = -1$ for the
   augmented root of $S$, we must have $\mu^{-1}(\rho_S) = \emptyset$.
   Thus, $\rho_S \succ_S \lca_S(\sT(v))$ for any $v\in V(T)$.  Hence, (D2)
   is trivially satisfied for $\rho_S$.  Moreover, $\tT(\rho_T) = 0$
   implies $\tT(u) > \ts(\rho_S)$ for any $u\in V(T)$. Hence, (D3) is
   always satisfied for $\rho_S$.

   In summary, Conditions (D1)-(D3) are met for any vertex $x\in V(S)$
   that up to this point has been assigned a value, i.e., $\ts(x)\ne\ast$.

   We will now assign to all vertices $x \in V(S)$ with $\mu^{-1}(x)
   =\emptyset$ a value $\ts(x)$ in a stepwise manner.  To this end, we give
   upper and lower bounds for the possible values that can be assigned to
   $\ts(x)$.  Let $x \in V(S)$ with $\ts(x) = \ast$.  Set 
\begin{eqnarray*}
  \LO(x) & = & \{ \tau_S(y)\ |\ x \prec_S y, y\in V(S) 
  \text{ and } \tau_S(y) \neq \ast \} \\
  \UP(x) & = & \{ \tau_S(y)\ |\
  x \succ_S y,  y \in V(S)  \text{ and } \tau_S(y) \neq \ast\}.\,
  \end{eqnarray*}
  We note that $\LO(x) \neq \emptyset$ and $\UP(x) \neq \emptyset$ because
  the root and the leaves of $S$ already have been assigned a value
  $\tau_S$ in the initial step.  In order to construct a valid time map
  $\ts$ we must ensure $\max(\LO(x)) < \tau_S(x) < \min(\UP(x))$. 
	
  Moreover, we strengthen the bounds as follows.  Put
  \begin{eqnarray*}
    \lo(x) & = & \{\tau_T(u)\ |\ t(u) \in \{\dpl, \hgt \} 
    \text{ and }  x \preceq_S
    \lca_S(\sT(u)) \} \\
    \up(x) & = & \{\tau_T(u)\ |\ \text{ where } (u,v)\in\mc{E}
    \text{ and } 
    \lca_S(\sT(u) \cup \sT(v)) \preceq_S x\ \}.\,
  \end{eqnarray*}
  Observe that $\max(\lo(x)) < \min(\up(x))$, since otherwise there are
  vertices $u, w \in V(T)$ with $\tT(w)\in \lo(x)$ and $\tT(u)\in \up(x)$
  and $\tT(w)\geq \tT(u)$. However, this implies that $\lca_S(\sT(u)\cup
  \sT(v)) \preceq_S x \preceq \lca_S(\sT(w))$; a contradiction to (T3).
	
  Since (D2) is satisfied for all vertices $y$ that obtained a value
  $\ts(y)\neq \ast$, we have $\max(\lo(x)) < \min(\UP(x))$. Likewise
  because of (D3), it holds that $\max(\LO(x)) < \min(\up(x))$. Thus we set
  $\ts(x)$ to an arbitrary value such that
  \begin{equation*}
    \max(\LO(x)\cup \lo(x)) < \ts(x) < \min(\UP(x)\cup \up(x)).
  \end{equation*} 
  By construction, (D1), (D2), and (D3) are satisfied for all vertices in
  $V(S)$ that have already obtained a time value distinct from $\ast$.
  Moreover, for all such vertices with $x\prec_T y$ we have
  $\ts(x)>\ts(y)$. In each step we chose a vertex $x$
	with $\ts(x) = \ast$ that obtains then a  real-valued time stamp. Hence, 
	in each step the number of vertices that have value $\ast$ is reduced by one.
	Therefore, repeating the latter procedure will eventually  assign to all vertices 
	a real-valued time stamp such that, in particular, $\ts$ satisfies (D1), (D2),
	  and (D3) and thus is indeed a time map for $S$.
\end{proof}

\section{Proof of Theorem \ref{thm:aux-C}}
\label{sec:thm:aux-C} 
\begin{proof} 
	Assume that $\mu$ is time-consistent. 
	By Theorem \ref{thm:D}, there are two time-maps $\tT$ and $\ts$ satisfying (C1) and (C2). 
	Let $\tau = \tT \cup \ts$ be the map from $V(T)\cup V(S)\to \mb{R}$. 
	Let $A'$ be the directed graph with $V(A') = V(S)
	\cup V(T)$ and set for all $x,y\in V(A')$: $(x,y)\in E(A')$ if and only if $\tau(x) < \tau(y)$. 
	By construction $A'$ is a DAG since $\tau$ provides a topological order on $A'$ \cite{Kahn:62}.
	
	We continue to show that $A'$ contains all edges of {$A_1$}.

	To see that (A1) is satisfied for $E(A')$ let $(u,v)\in E(T)$. Note, 
	$\tau(v)>\tau(u)$, since $\tT$ is a  time map for $T$ and by construction of $\tau$. 
	Hence, all edges $(u,v)\in E(T)$ are also contained in $A'$, independent from the respective event-labels
	$t(u), t(v)$. Moreover, if $t(u)$ or $t(v)$ are speciation vertices or leaves, then 
	(C1) implies that $\tau_S(\mu(u)) =	\tau_T(u) > \tau_T(v)$ or 
	$\tau_T(u) > \tau_T(v) = \tau_S(\mu(v))$. By construction of $\tau$, 
	all edges satisfying (A1) are contained in $E(A')$.
	Since $\ts$ is a  time map for $S$, all edges as in (A2) are contained in $E(A')$. 
	Finally, (C2) implies that all edges satisfying (A5)  are contained in $E(A')$.
	
	Although, $A'$ might have more edges than required by (A1), (A2) and (A5), 
	the graph $A_1$ is a subgraph of $A'$. Since  $A'$ is a DAG, also $A_1$ is a DAG.

	For the converse assume that $A_1$ is a directed graph with $V(A_1) =
	V(S) \cup V(T)$ and edge set $E(A_1)$ as constructed in  Def.\ \ref{def:auxG} 
	(A1), (A2) and (A5).
	Moreover, assume that $A_1$ is a DAG. Hence, there is 
 	is a topological order $\tau$ on $A_1$ with $\tau(x) < \tau(y)$ whenever $(x,y) \in E(A_1)$.
	In what follows we construct the time-maps $\tT$ and $\ts$
	such that they satisfy (C1) and (C2). Set $\ts(x) = \tau(x)$ for all $x\in V(S)$.
	Additionally, set for all $u \in V(T)$:
	\[
		\tau_T(u) = 
		\begin{cases}
			\tau(\mu(u)) & \text{if } t(u) \in \{ \lea, \spe \}\\
			\tau(u) & \text{otherwise}.\,
		\end{cases}
	\]
	By construction it follows that (C1) is satisfied. Due to (A2), $\tau_S$ is
	a valid time map for $S$. It follows from the construction and (A1) that $\tau_T$ is a
	valid time map for $T$. 
	Assume now that $u\in V(T)$, $t(u) \in \{ \dpl, \hgt \}$, and 
	$\mu(u) = (x,y) \in E(S)$. Since $\tau$ provides  a topological
	order we have:
	\[
		\tau(x) \stackrel{(A5)}{<} \tau(u) \stackrel{(A5)}{<} \tau(y).\,
	\] By construction, it follows that $\ts(x) < \tT(u) < \ts(y)$ satisfying (C2).
	
\end{proof}

\section{Proof of Theorem \ref{thm:aux}}
\label{sec:thm:aux}
\begin{proof} 
	The proof is similar to the proof of Theorem \ref{thm:aux-C}.
	Assume there is a time-consistent reconciliation map $(T;t, \sigma)$ to $S$. 
	By Theorem \ref{thm:D},  there are two time-maps $\tT$ and $\ts$ satisfying (D1)-(D3). 
	Let
	$\tau$ and $A'$ be defined as in the proof of Theorem \ref{thm:aux-C}.

	Analogously to the proof of Theorem \ref{thm:aux-C}, we show that $A'$
	contains all edges of $A_2$. Application of (D1) immediately implies that all
	edges satisfying (A1) and (A2) are contained in $E(A')$. By condition (D2),
	it yields $(u,lca_S(\sT(u))) \in E(A')$ and (D3) implies $(lca_S(\sT(u)\cup
	\sT(v)),u)\in E(A')$. We conclude by the same arguments as before that the
	graph $A_2$ is a DAG.
	
	For the converse, assume we are given the directed acyclic graph $A_2$.
	As before, there is 
 	is a topological order $\tau$ on $A_2$ with $\tau(x) < \tau(y)$ only if $(x,y) \in E(A_2)$.
	The time-maps $\tT$ and $\ts$ are given as in the proof of Theorem \ref{thm:dag125}.

	By construction, it follows that (D1) is satisfied. Again, by construction and
	the Properties (A1) and (A2), 
	$\tau_S$  and $\tau_T$ are  valid time-maps for $S$ and $T$ respectively.
	
	Assume now that $u\in V(T)$, $t(u) \in \{ \dpl, \hgt \}$, and $x \preceq_S
	\lca_S(\sT(u))$ for some $x\in V(S)$. Since there is a topological
	order on $V(A_2)$, we have
	\[
		\tau(x) \stackrel{(A2)}{\geq} \tau(\lca_S(\sT(u))) \stackrel{(A3)}{>} \tau(u).\,
	\]
	By construction, it follows that $\ts(x) > \tT(u)$. Thus, (D2) is
	satisfied.

	Finally assume that $(u,v)\in \mc{E}$ and $\lca_S(\sT(u) \cup \sT(v))
	\preceq_S x$ for some $x \in V(S)$. Again, since $\tau$ provides a topological
	order, we have:
	\[
		\tau(x) \stackrel{(A2)}{\leq} \tau(\lca_S(\sT(u) \cup \sT(v))) \stackrel{(A4)}{<} \tau(u).\,
	\] By construction, it follows that $\ts(x) < \tT(u)$, satisfying (D3).

	Thus $\tau_T$ and $\tau_S$ are valid time maps satisfying (D1)-(D3).
\end{proof}

\section{Algorithm \ref{alg:compute-sigma}, Proof of Lemma \ref{lem:lca} and Theorem \ref{thm:main-algo}}
\label{proof-Algos}

\begin{algorithm}
\caption{Compute $\ell(u) = \lca_S(\sT(u))$ for all $u \in V(T)$
\label{alg:compute-sigma}}
\begin{algorithmic}[1]
\Function{ComputeLcaSigma}{$(T;t,\sigma), S$}
	\State{$\ell(u) \gets \emptyset$ for all $u\in V(T)$} \Comment{``$\emptyset$'' means uninitialized}
	\Let{$A$}{empty stack} 
	\State $A.push(\rho_T)$
	\While{$A$ is not empty}
		\Let{$u$}{$A.pop()$}
		\If{$t(u) = \lea$} $\ell(u)\gets \sigma(u)$  \label{alg:lca:leaf}
		\ElsIf{$\ell(v) = \emptyset$ for some child $v$ of $u$} $A.push(u)$, $A.push(v)$ \label{alg:lca:child}
 		\Else \label{alg:lca:child2}
			\Let{$\ell(u)$}{$\lca_S(\{ \ell(v)\ |\ (u,v)\in E(T) $ and $ t((u,v)) = 0 \})$} \label{alg:lca:child3}
		\EndIf
	\EndWhile
	\State\Return $\ell$
\EndFunction
\end{algorithmic}
\end{algorithm}

\begin{proof}[Proof of Lemma \ref{lem:lca}] 
	Let $u\in V(T)$. In what follows, we show that $\ell(u) =\lca_S(\sT(u))$.  
   In fact, the algorithm is (almost) a depth first search through $T$ that 
	assigns the (species tree) vertex $\ell(u)$ to $u$ if and only if
	every child $v$ of $u$ has obtained an assignment $\ell(v)$
	(cf.\ Line \eqref{alg:lca:child2} - \eqref{alg:lca:child3}). 
	That there are children $v$ with non-empty 
   $\ell(v)$ at some point is ensured by Line \eqref{alg:lca:leaf}. 
	That is, if  $t(u) = \lea$, then $\ell(u) = \lca_S(\sT(u)) = \sigma(u)$. 
	Now, assume there is an interior vertex $u\in V(T)$, where every child $v$
	has been assigned a value $\ell(v)$, then
	\begin{align*}
		\lca_S(\sT(u)) =
		&\lca_S(\sT(\{ \sT(v)\ |\ (u,v)\in E(T) \text{ and } t(u,v) = 0 \}))  \\
		=&\lca_S(\sT(\{ \lca_S(\sT(v))\ |\ (u,v)\in E(T) \text{ and } t(u,v) = 0 \}))  \\
		=&\lca_S(\sT(\{ \ell(v)\ |\ (u,v)\in E(T) \text{ and } t(u,v) = 0 \})) 
	\end{align*}
	The latter is achieved by Line \eqref{alg:lca:child3}. 

	Since $T$ is a tree and the algorithm is in effect a depth first search through
	$T$, the while loop runs at most $O(V(T)+ E(T))$ times, and thus in $O(V(T))$ time. 

	The only non-constant operation within the while loop is the computation of
	$\lca_S$ in Line $\eqref{alg:lca:child3}$. Clearly $\lca_S$ of a set of
	vertices $C = \{ c_1, c_2 \dots c_k \}$, where $c_i \in V(S)$, for all $c_i
	\in C$ can be computed as sequence of $\lca_S$
	operations taking two vertices: $\lca_S(c_1,\lca_S(c_2, \dots
	\lca_S(c_{k-1}, c_k)))$, each taking $O(\lg(|V(S)|))$ time. Note however,
	that since Line \eqref{alg:lca:child3} is called exactly once for each
	vertex in $T$, the number of $\lca_S$ operations taking two vertices is
	called at most $|E(T)|$ times through the entire algorithm. Hence, the total
	time complexity is $O(|V(T)|\lg(|V(S)|))$.
\end{proof}

\begin{proof}[Proof of Theorem \ref{thm:main-algo}]
	In order to produce a time-consistent reconciliation map, we first
	construct some valid reconciliation map $\mu$ from $(T;t,\sigma)$ to $S$. Using
	the $\lca$-map $\ell$ from Algorithm \ref{alg:compute-sigma}, $\mu$ will be
	adjusted to become time-consistent, if possible. 
		
	By assumption, there is  a reconciliation map from $(T;t,\sigma)$ to $S$.
	The for-loop (Line \eqref{alg:all:for1}-\eqref{alg:all:for2}) ensures that each 
	vertex $u\in V$ obtained a value $\mu(u)$. 
	We continue to show that $\mu$ is a valid reconciliation map satisfying 
	(M1)-(M3). 

	Assume that $t(u) = \lea$, in this case $\ell(u) = \sigma(u)$, and thus
	(M1) is satisfied. If $t(u) = \spe$, it holds that $\mu(u) = \ell(u) =
	\lca_S(\sT(u))$, thus satisfying (M2i). Note that $\rho_S \succ_S \ell(u)$, 
	and hence, $\mu(u) \in F$ by Line \eqref{alg:all:edge}, implying that
	(M2ii) is satisfied.
	Now, assume $t(u) = \hgt$ and $(u,v)\in \mc{E}$. By assumption, we know
	there exists a reconciliation map from $T$ to $S$, thus by ($\Sigma$2):
	\[	\sT(u) \cap \sT(v) = \emptyset	\]
	It follows that, $\ell(u)$ is incomparable to $\ell(v)$, satisfying (M2iii).

	Now assume that $u,v\in V$ and $u \prec_{\Th} v$. Note that
	$\sT(u) \subseteq \sT(v)$. It follows that
	$\ell(u) = \lca_S(\sT(u)) \preceq_S \lca_S(\sT(v)) = \ell(v)$.
	By construction,  (M3) is satisfied. Thus, $\mu$ is a valid reconciliation map.

	By Theorem \ref{thm:aux}, two time maps $\tau_T$ and $\tau_S$ satisfying (D1)-(D3) only exists if
	the auxiliary graph $A$ build on Line \eqref{alg:all:A} is a DAG. Thus if $A$ contains a
	cycle, no such time-maps exists and the statement ``No time-consistent reconciliation map exists.''
   is returned (Line \eqref{alg:all:cycle}). On the other hand, if $A$ is a DAG, the
	construction in Line \eqref{alg:all:constt1}-\eqref{alg:all:constt4} is 
	identical to the construction used in the
	proof of Theorem \ref{thm:aux}. Hence correctness of this part of the algorithm 
	follows directly from the proof of Theorem \ref{thm:aux}.
	
	Finally, we adjust $\mu$ to become a time-consistent reconciliation map.. 
	By the latter arguments, $\tau_T$ and $\tau_S$ satisfy (D1)-(D3) w.r.t.\
	to $\mu$. Note, that $\mu$ is chosen to be the ``lowest point'' where 
	a vertex $u \in V$ with $t(u) \in \{ \dpl, \hgt \}$ can be mapped, that is, 
	$\mu(u)$ is set to $(p(x), x)$ where $x =\lca_S(\sT(u))$.
	However, by the arguments in the proof of Theorem	\ref{thm:D}, there is 
	a unique edge $(y,z)\in	W$ on the path from $x$ to $\rho_S$ such that $\tau_S(y) < \tau_T(u) < \tau_S(z)$.
	The latter is ensured by choosing a different value for distinct vertices in $V(A)$, see comment in Line \eqref{alg:all:comment}.
	Hence, Line \eqref{alg:all:mu} ensures, that $\mu(u)$ is mapped on the correct edge such that 
	(C2) is satisfied. 
	It follows that adjusted $\mu$ is a valid time-consistent reconciliation map.
	
	We are now concerned with the time-complexity. 
	By Lemma \ref{lem:lca}, computation of 	 $\ell$ in Line \eqref{alg2:lca} takes 
	 $O(|V|\log(|W|))$ time and the for-loop  (Line \eqref{alg:all:for1}-\eqref{alg:all:for2})
	takes $O(|V|)$ time. 
	We continue to show that the auxiliary 
	graph $A$ (Line \eqref{alg:all:A})  can be constructed in $O(|V|\log(|W|))$ time.

	Since we know $\ell(u) = \lca_S(\sT(u))$ for all $u\in V$ and since $T$ and $S$
   are trees, the subgraph with edges satisfying (A1)-(A3) can be constructed in 
	$O(|V|+|W| + |E| + |F)|) = O(|V|+|W|)$ time. 
	To ensure (A4),  we must compute for a possible transfer edges $(u,v)\in \EHGT$ 
	the vertex $\lca_S(\sT(u) \cup	\sT(v))$. 
	which can be done in $O(\log(|W|))$ time. 
	Note, the number of transfer edges is bounded by the number of possible transfer event $O(|V|)$. 
	Hence, generating all edges satisfying (A4) takes 
	 $O(|V|(\log(|W|))$ time. 
	In summary, computing $A$ can done in $O(|V|+|W|+|V|(\log(|W|)) = O(|V|(\log(|W|))$ time. 
	
	To detect whether $A$ contains cycles one has to determine whether there is 
	a topological order $\tau$ on $V(A)$ which can be done 
	via depth first search in $O(|V(A)|+|E(A)|)$ time.
	Since $|V(A)| = |V|+|W|$ and $O(|E(A)|) = O(|F|+|E|+|W|+|V|)$	and $S,T$ are trees, 
	the latter task can be done in $O(|V| + |W|)$ time. 
	Clearly, Line \eqref{alg:all:constt2}-\eqref{alg:all:constt4} can be performed
	on $O(|V| + |W|)$ time.  
	
	Finally, we have to adjust $\mu$ according to $\tT$ and $\ts$. 
	Note, that for each 	$u\in V$ with 
	$t(u) \in \{ \dpl, \hgt \}$ (Line \eqref{alg:all:forx}) 
  	we have possibly adjust $\mu$ to the next edge $(p(x),x)$. 
	However, the possibilities for the choice of $(p(x),x)$
	is bounded by by the height of $S$, 	which is in the worst case  
	$\log(|W|)$.  Hence, the for-loop in Line \eqref{alg:all:forx}
	has total-time complexity  $O(|V|\log(|W|))$.
	
	In summary, the overall time complexity of Algorithm \ref{alg:all} is 
	$O(|V|\log(|W|))$.
\end{proof}

\end{document}